\documentclass[reqno,9pt]{amsart}

\usepackage{amsthm, amsmath, amssymb, amsfonts, graphicx}

\usepackage[OT4]{fontenc}
\usepackage[utf8]{inputenc}
\usepackage[english]{babel}

\usepackage{bbm}

\usepackage[inline]{enumitem}

\usepackage{float}
\restylefloat{table}
\usepackage{lipsum}

\usepackage[colorlinks=true, pdfstartview=FitV, linkcolor=blue,
            citecolor=blue, urlcolor=blue]{hyperref}
\usepackage[usenames]{color} 

\usepackage[round]{natbib}

\usepackage{graphicx}
\usepackage[font=sl,labelfont=bf,width=\textwidth]{caption}
\usepackage{color}

\newcommand{\wo}{\mathbb{E}}
\DeclareMathOperator{\cov}{Cov}
\DeclareMathOperator{\cor}{Cor}

\usepackage{multirow}
\usepackage{amsthm}
\newtheorem{theorem}{Theorem}

\newtheorem{example}[theorem]{Example}
\newtheorem{remark}[theorem]{Remark}

\def\bP{\mathbb{P}}

\def\bN{\mathbb{N}}
\def\cF{\mathcal{F}}
\usepackage{bbm}

\usepackage{amssymb}
\usepackage{placeins} 
\usepackage{graphicx}

\usepackage{setspace}
\usepackage[foot]{amsaddr}

\graphicspath{{Images/}}

\begin{document}

\title[Conditional correlation estimation and serial dependence identification]{Conditional correlation estimation and serial dependence identification}

\author{Kewin P\k{a}czek$^{\ast}$}
\author{Damian Jelito$^{\ast}$}
\author{Marcin Pitera$^{\ast}$}
\author{Agnieszka Wy\l{}oma\'{n}ska$^{\dagger}$}
\address{$^{\ast}$Institute of Mathematics, Jagiellonian University, S. {\L}ojasiewicza 6, 30-348 Krak{\'o}w, Poland}
\address{$^{\dagger}$Faculty of Pure and Applied Mathematics, Hugo Steinhaus Center, Wroc{\l}aw University of Science and Technology, Wyspia{\'n}skiego 27, 50-370 Wroc{\l}aw, Poland}
\email{kewin.paczek@im.uj.edu.pl, damian.jelito@uj.edu.pl, marcin.pitera@uj.edu.pl, agnieszka.wylomanska@pwr.edu.pl}

\maketitle

\vspace{-1cm}
\begin{abstract}
It has been recently shown in Jaworski, P., Jelito, D. and Pitera, M. (2024), ‘A note on the equivalence between the conditional uncorrelation and the independence of
random variables’, Electronic Journal of Statistics 18(1),
that one can characterise the independence of random variables via the family of conditional correlations on quantile-induced sets. This effectively shows that the localized linear measure of dependence is able to detect any form of nonlinear dependence for appropriately chosen conditioning sets. In this paper, we expand this concept, focusing on the statistical properties of conditional correlation estimators and their potential usage in serial dependence identification. In particular, we show how to estimate conditional correlations in generic and serial dependence setups, discuss key properties of the related estimators, define the conditional equivalent of the autocorrelation function, and provide a series of examples which prove that the proposed framework could be efficiently used in many practical econometric applications. 
\vspace{0.2cm}

\end{abstract}



\section{Introduction}

The linear Person's correlation coefficient is arguably the most recognized metric used for identification of statistical dependence. This applies to many disciplines, including statistics, time series analysis, pattern recognition, and signal processing, see \cite{BenCheHua2008} and references therein.  In particular, for serial dependence identification, autocovariance and autocorrelation is the core diagnostic tool used in stationary signal analysis, see e.g.~\cite{brockwell2016introduction}.

Although linear correlation is defined as normalized covariance, it has multiple interpretations linked to regression, scoring, and variance ratio evaluation, which makes this metric very appealing to practitioners, see e.g.~\cite{LeeNic1988}. That being said, statistical analysis relying on correlation measurements can be perplexing and result in subtle mistakes if not handled carefully, see, e.g.,  \cite{KotDro2001,Ald1995,Bel2015}. Namely, while correlation is designed to assess the linear relationship between random variables or within time series, it frequently does not adequately capture nonlinear dependencies. Although the copula function can fully describe the dependence, practitioners often prefer to use simpler characteristics, such as correlation, to measure the degree of dependence, see~\cite{Nel2006}. 

To remediate that, many alternative measures have been proposed in the literature and this field is constantly evolving, see e.g. ~\cite{Nel2006,Sca1984,RaoSetXuCheTagPri2011,ZhuXuRunZho2017,AkeBruRobSimWei1984,BabShiSib2004,Gre2005,KenHuaVodHavSta2015} for the definitions of measures of association, concordance measures, entropy correlations, projection correlations, or tail correlations. Recently, in \cite{JawJelPit2022}, it has been shown that one can fully characterize the independence of random variables via the family of conditional correlations on quantile-induced sets; see also~\cite{Tar2024}. This effectively shows that the localized version of the simple correlation coefficient is able to detect any form of non-linear dependence, for appropriately chosen conditioning set. These results complement a series of recent papers focused on conditional moment analysis and their applications, see \cite{JawPit2015,PitCheWyl2021,PacJelPitWyl2022,JawPit2020,HebZimPitWyl2019} for details. Furthermore, in Section 4 of \cite{JawJelPit2022}, a discussion on potential applications of conditional correlation has been presented, and it has been stated that conditional correlations can be used for efficient independence testing and statistical detection of serial dependence.

In this paper, we build on the concept of {\it Quantile Conditional Correlation} (QCC) introduced in \cite{JawJelPit2022}, see~\eqref{eq:cond_cor} for the definition. We focus on the statistical properties of conditional correlation estimators and their potential usage in time series analysis. Our paper also expands the statistical results from~\cite{JelPit2018}, where the properties of univariate conditional variance estimators have been studied, as well as from~\cite{WozJawJelPitWyl2024}, where the properties of conditional sample covariance have been studied under a different univariate conditioning scheme based on a predefined benchmark random variable. In particular, we show how to estimate conditional correlations in generic and serial dependence setups, discuss key properties of the related estimators, and provide a series of examples which prove that this tool could be efficient in many practical econometric frameworks. Note that while the focus of the paper~\cite{JawJelPit2022} is on the theoretical probabilistic properties of QCC, here we show how to estimate QCC and apply it for statistical serial dependence identification.

From a practical statistical perspective, the utilization of QCC offers several advantages over the standard correlation, particularly in the context of identifying serial dependence. Firstly, QCC serves as a natural and robust alternative to classical measures when dealing with data from heavy-tailed distributions. In such cases, the classical measure, namely the Pearson correlation coefficient, is inadequate as it becomes undefined for infinite-variance random variables. Furthermore, as demonstrated in the current paper, QCC can form the basis for autodependence measures and can effectively replace the classical autocovariance function (ACVF) or autocorrelation function by providing their robust equivalents. Autodependence measures are critical in time series analysis because they underpin techniques for theoretical model identification and parameter estimation, see, e.g., \cite{brockwell2016introduction}. Robust techniques are particularly beneficial when dealing with data exhibiting heavy-tailed (impulsive) behaviour. This issue was highlighted in \cite{ZULAWINSKI2024111367}, where a model with heavy-tailed behaviour was discussed in the context of condition monitoring (vibration-based local damage detection), and robust measures were proposed for classical ACVF-based algorithms. The QCC-based approach can also be valuable in this context. Additionally, the approach proposed in this paper is shown to be useful for noise-corrupted models. Models with external additive noise, which may also exhibit impulsive behaviour, are extensively discussed in the literature, see, e.g., \cite{esfandiari,parma_ao0,BarbieriDiversi+2024+65+79,doi:10.1080/02331888.2015.1060237}. External noise is a natural consequence of external factors influencing observations in industrial applications or influences from other markets in financial time series. The presence of external noise significantly diminishes the efficiency of classical methods. The QCC-based methodology meets the demands of contemporary data analysis, as the issues of heavy-tailed behaviour and external noise in data are increasingly highlighted by researchers across various fields, see, e.g., \cite{PhysRevE.91.062809,7054546} and \cite{9716126,doi:10.1080/00207721.2021.2012726,10.1111/j.1467-9868.2011.00775.x}.

This paper is organized as follows. In Section~\ref{S:preliminaries}, we set up the terminology and recall the selected theoretical results of conditional correlations. Then, in Section~\ref{S:cond_cor_est} we introduce the statistical setup, show how to estimate conditional correlations, and present properties of the sample estimators. Next, in Section~\ref{S:cond_autocor} we transfer the results from Section~\ref{S:cond_cor_est} to the serial dependence identification setup, define a conditional version of the autocorrelation function, and prove its basic properties. Finally, in Section~\ref{s4_intro} we show applications of conditional correlations to serial dependence detection by performing power simulation studies and analyzing conditional autocorrelations of S\&P500 stock market data.

\section{Preliminaries}
\label{S:preliminaries}
Let $(\Omega,\cF,\bP)$ be a probability space and let $(X,Y)$ be a generic $\mathbb{R}^2$-valued random vector defined on that space. We use $F$, $F_X$, and $F_Y$ to denote the distribution functions of $(X,Y)$, $X$, and $Y$, respectively. We assume that $F_X$ and $F_Y$ are continuous and invertible, and use $Q_X:=F_X^{-1}$ and $Q_Y:=F_Y^{-1}$ to denote the inverse (quantile) functions of $X$ and $Y$, respectively. Also, for any $B\in \mathcal{F}$, we use $1_B$ to denote the characteristic function of $B$; notation $\overset{\mathbb{P}} {\longrightarrow}$ and $\overset{d}{\longrightarrow}$ is used to denote convergence in probability and in distribution, respectively. 

Throughout this section we fix the quantile splits $0<p_1<q_1<1$ and $0<p_2<q_2<1$ and define the corresponding $(X,Y)$ quantile conditioning set
\begin{equation}\label{eq:set_A}
A:= [Q_X(p_1),Q_X(q_1)] \times [Q_Y(p_2),Q_Y(q_2)].
\end{equation}
For simplicity, we assume that the splits are chosen in such a way that $\bP[(X,Y)\in A]>0$; note this is always the case if $(X,Y)$ has full support. With this split, we associate the quantile conditional mean, variance, and covariance given by
\begin{align}
    \mu_A(Z)&:=\wo[Z|(X,Y)\in A], \label{eq:cond_exp}\\
    \sigma^2_A(Z) &:= \wo[(Z-\mu_A(Z))^2|(X,Y)\in A],\label{eq:cond_var}\\
    \cov_A(X,Y) &:= \wo[\left(X-\mu_A(X)\right) \left(Y-\mu_A(Y)\right)|(X,Y)\in A],\label{eq:cond_cov}
\end{align}
respectively, where $Z$ is some generic random variable on $(\Omega, \mathcal{F},\mathbb{P})$; typically, we consider $Z=X$ or $Z=Y$. Note that $\mu_A(X)$, $\mu_A(Y)$, $\sigma^2_A(X)$, $\sigma^2_A(Y)$, and $\cov_A(X,Y)$ are conditional versions of standard moment-related characteristics of the bivariate distribution of $(X,Y)$, where the conditioning refers to the space-trimming. Next, we define the quantile conditional correlation (QCC) given by
\begin{equation}\label{eq:cond_cor}
\cor_A({X},{Y}):=\frac{\cov_A({{X}},{{Y}})}{\sqrt{\sigma^2_A({X})\sigma^2_A({Y})}}.
\end{equation}
Note that \eqref{eq:cond_cor} always exists and is finite, even if the underlying unconditioned random variables are not square integrable; this follows from the fact that the quantile splits satisfy $0<p_1<q_1<1$ and $0<p_2<q_2<1$. Also, in the limit case when $p_1=p_2=0$ and $q_1=q_2=1$, we recover the standard unconditional correlation (if it exists). As in the standard unconditional setting, the conditional correlation remains invariant under positive transformations of the random vector $(X,Y)$; although  $A$ might change, the set $\{(X,Y) \in A\}$ remains the same.

The usefulness of the conditional correlations is justified by Theorem 1 from~\cite{JawJelPit2022} which is recalled below for clarity. In a nutshell, the theorem states that the family of conditional correlations fully characterizes the independence of $X$ and $Y$.

\begin{theorem}\label{th:independence}
    Random variables $X$ and $Y$ are independent if and only if they are conditionally uncorrelated on every
quantile set, i.e. for any quantile splits $0 <p_1 < q_1 <1 $ and $0 <p_2 < q_2 <1 $ and the related set $A$, we get $\cor_A(X,Y) = 0$.
\end{theorem}

 Theorem~\ref{th:independence} indirectly shows that conditional moments could be used to identify the distribution or to extract its characteristics. This was recently used to design some new estimation and goodness-of-fit procedures for heavy-tailed data; see e.g~\cite{PitCheWyl2021} or \cite{PacJelPitWyl2022}. In the next section, we will demonstrate how to use the sample version of~\eqref{eq:cond_cor} to statistically assess the independence of random samples.

\begin{remark}[Closed-form formulas for conditional moments]
The analytic conditional moment formulas for univariate truncations and quantile conditionings are provided in the literature for multiple families of distributions, see e.g.~\cite{Nad2009}. Their bivariate counterparts are typically harder to compute and are available only for selected types of truncations and families. That saying, it is worth mentioning that the explicit formula for conditional moments of the multivariate normal vector can be found in \cite{Manjunath2021}, see formulas (17) and (18) therein.
\end{remark}


\section{Estimation of conditional correlations}\label{S:cond_cor_est}

In this section, we show how to estimate the quantile conditional correlation  and discuss its properties. Throughout this section, for any $n\in\bN\setminus \{0\}$, we assume that we are given an independent and identically distributed (i.i.d.) bivariate sample from $(X,Y)$ and denote it by
\begin{equation}\label{eq:sample}
(\mathbf{X},\mathbf{Y}):=(X_i,Y_i)_{i=1}^n.
\end{equation}
Similarly, for an arbitrary univariate vector $Z$ and $n\in\bN\setminus \{0\}$, we use $\mathbf{Z}:=(Z_1, \ldots, Z_n)$ to denote its sample. We maintain the dependence structure, i.e. $(\mathbf{X},\mathbf{Y},\mathbf{Z})$ is effectively a sample from $(X,Y,Z)$; in fact, we typically set $\mathbf{Z}=\mathbf{X}$ or $\mathbf{Z}=\mathbf{Y}$, so that the information contained in the original sample \eqref{eq:sample} is sufficient. For an arbitrary sample $\mathbf{Z}$, we use $Z_{(k)}$ to denote the corresponding $k$th order statistic (i.e. the $k$th smallest element in the sample), for $k=1, \ldots, n$. Given a sample $(\mathbf{X},\mathbf{Y})$ as well as quantile splits $0<p_1<q_1<1$ and $0<p_2<q_2<1$, we define the empirical equivalent of the conditioning set $A$ given in \eqref{eq:set_A.hat}, by setting 
\begin{equation}\label{eq:set_A.hat}
\hat{A} :=[X_{([np_1]+1)},X_{([nq_1])}] \times [Y_{([np_2]+1)},Y_{([nq_2])}],
\end{equation}
where $[a]:=\sup\{k\in \mathbb{Z}\colon k\leq a\}$ is the integer part of $a\in \mathbb{R}$. Note that $\hat{A}$ could be seen as a sample version of the set $A$ with random rectangle vertices given via the standard empirical quantile estimators applied separately to each coordinate. Using \eqref{eq:set_A.hat}, we define sample empirical equivalents of \eqref{eq:cond_exp}, \eqref{eq:cond_var}, and \eqref{eq:cond_cov}. Namely, we set
\begin{align}
    \textstyle\hat\mu_{A}(\mathbf{Z})&\textstyle :=\frac{1}{\sum_{i=1}^n 1_{\hat{A}}(X_i,Y_i) } \sum_{i=1}^n Z_i 1_{\hat{A}}(X_i,Y_i), \label{eq:cond_exp_emp}\\
    \textstyle\hat\sigma^2_A(\mathbf{Z}) &\textstyle:= \frac{1}{\sum_{i=1}^n 1_{\hat{A}}(X_i,Y_i)} \sum_{i=1}^n (Z_i - \hat{\mu}_{A}(\mathbf{Z}))^2  1_{\hat{A}}(X_i,Y_i),\label{eq:cond_var_emp}\\
    \textstyle\widehat\cov_A(\mathbf{X},\mathbf{Y}) &\textstyle:= \frac{1}{\sum_{i=1}^n 1_{\hat{A}}(X_i,Y_i)} \sum_{i=1}^n (X_i - \hat{\mu}_{A}(\mathbf{X})) (Y_i - \hat{\mu}_{A}(\mathbf{Y})) 1_{\hat{A}}(X_i,Y_i),\label{eq:cond_cov_emp}
\end{align}
with the convention that $\hat\mu_{A}(\mathbf{Z})=\hat\sigma^2_A(\mathbf{Z})=\widehat\cov_A(\mathbf{X},\mathbf{Y})=0$ if $\sum_{i=1}^n 1_{\hat{A}}(X_i,Y_i)=0$. Recall that the splits have been chosen in such a way that $\bP[(X,Y)\in A]>0$, which implies that $\bP[\sum_{i=1}^n 1_{\hat{A}}(X_i,Y_i)>0] \nearrow 1$ as $n\to \infty$. Also, recall that for $\mathbf{Z}=\mathbf{X}$ or $\mathbf{Z}=\mathbf{Y}$, the information contained in \eqref{eq:sample} is sufficient to estimate \eqref{eq:cond_exp_emp}, \eqref{eq:cond_var_emp}, and \eqref{eq:cond_cov_emp}. Still, even for $\hat\mu_{A}(\mathbf{X})$ (or $\hat\mu_{A}(\mathbf{Y})$) we need information encoded in both margins of $(\mathbf{X},\mathbf{Y})$, as the set $\hat A$ is a double quantile-truncated set. Next, using the same logic as in \eqref{eq:cond_cor}, we define the sample conditional correlation estimator 
\begin{equation}\label{eq:cond_cor_est}
\widehat\cor_A(\mathbf{X},\mathbf{Y}):=\frac{\widehat\cov_A({\mathbf{X}},{\mathbf{Y}})}{\sqrt{\hat\sigma^2_A(\mathbf{X})\hat\sigma^2_A(\mathbf{Y})}},
\end{equation}
with the convention $\widehat\cor_A(\mathbf{X},\mathbf{Y})=0$ if $\sum_{i=1}^n 1_{\hat{A}}(X_i,Y_i)=0$.

Now, let us discuss selected asymptotic properties of~\eqref{eq:cond_cor_est}.  For brevity, we introduce supplementary notation
\begin{align}
    \textstyle\bar\mu_{A}(\mathbf{Z})&\textstyle:=\frac{1}{\sum_{i=1}^n 1_{{A}}(X_i,Y_i) } \sum_{i=1}^n Z_i 1_{{A}}(X_i,Y_i), \label{eq:cond_exp_emp_known}\\
    \textstyle\bar\sigma^2_A(\mathbf{Z}) &\textstyle:= \frac{1}{\sum_{i=1}^n 1_{{A}}(X_i,Y_i)} \sum_{i=1}^n (Z_i - \bar{\mu}_{A}(\mathbf{Z}))^2  1_{{A}}(X_i,Y_i),\label{eq:cond_var_emp_known}\\
    \textstyle\overline\cov_A(\mathbf{X},\mathbf{Y}) & \textstyle:= \frac{1}{\sum_{i=1}^n 1_{{A}}(X_i,Y_i)} \sum_{i=1}^n (X_i - \bar{\mu}_{A}(\mathbf{X})) (Y_i - \bar{\mu}_{A}(\mathbf{Y})) 1_{{A}}(X_i,Y_i)\label{eq:cond_cov_emp_known}
\end{align}
with the same convention as in~\eqref{eq:cond_exp_emp}--\eqref{eq:cond_cov_emp}. Note that estimators  \eqref{eq:cond_exp_emp_known}--\eqref{eq:cond_cov_emp_known} are versions of \eqref{eq:cond_exp_emp}--\eqref{eq:cond_cov_emp} in which we replaced the (random) set $\hat A$ with $A$, i.e. these modified estimators pre-assume full-knowledge about the quantile split structure. Similarly, we define
\begin{equation}\label{eq:cond_cor_est_known}
\overline\cor_A(\mathbf{X},\mathbf{Y}):=\frac{\overline\cov_A({\mathbf{X}},{\mathbf{Y}})}{\sqrt{\overline\sigma^2_A(\mathbf{X})\overline\sigma^2_A(\mathbf{Y})}}.
\end{equation}
Essentially, by considering estimators \eqref{eq:cond_exp_emp_known}--\eqref{eq:cond_cor_est_known} we can split~\eqref{eq:cond_exp_emp}--\eqref{eq:cond_cor_est} into the component with deterministic sum and the residual component; this split is a common setup in the conditional moment analysis, see e.g.~\cite{Stigler1973} and~\eqref{eq:th:consistency:1} in this paper for details. 

We are now ready to show the consistency property.

\begin{theorem}\label{th:consistency}
    Let $\cor_A(X,Y)$ and $\widehat\cor_A(\mathbf{X},\mathbf{Y})$ be given by~\eqref{eq:cond_cov} and~\eqref{eq:cond_cov_emp}, respectively.  Then, we get
    \[
    \widehat\cor_A(\mathbf{X},\mathbf{Y})\overset{\mathbb{P}}{\longrightarrow} \cor_A(X,Y), \quad n\to \infty.
    \]
\end{theorem}
\begin{proof}
    First, note that it is enough to show $ \widehat\cov_A(X,Y)\overset{\mathbb{P}}{\to} \cov_A(X,Y)$, $ \hat\sigma^2_A(X)\overset{\mathbb{P}}{\to}\sigma^2_A(X)$, and $\hat\sigma^2_A(Y)\overset{\mathbb{P}}{\to}\sigma^2_A(Y)$  as $n\to\infty$ and use the continuous mapping theorem. In fact, we focus on the consistency of the sample conditional covariance; the remaining results could be shown using similar techniques.
Thus, let us show that $\widehat\cov_A(\mathbf{X},\mathbf{Y})\overset{\mathbb{P}}{\longrightarrow} \cov_A(X,Y)$. Noting that
\begin{equation}\label{eq:th:consistency:0}
        \widehat\cov_A(\mathbf{X},\mathbf{Y}) = \frac{n}{\sum_{i=1}^n 1_{\hat{A}}(X_i,Y_i)} \frac{1}{n}\sum_{i=1}^n X_i Y_i 1_{\hat{A}}(X_i,Y_i)- \hat{\mu}_{A}(\mathbf{X})\hat{\mu}_{A}(\mathbf{Y}),
\end{equation}
it is enough to show that $\frac{1}{n}\sum_{i=1}^n X_i Y_i 1_{\hat{A}}(X_i,Y_i)\overset{\mathbb{P}}{\to} \mathbb{E}[XY 1_A(X,Y)]$, $\frac{1}{n}\sum_{i=1}^n 1_{\hat{A}}(X_i,Y_i)\overset{\mathbb{P}}{\to} \mathbb{P}[(X,Y)\in A]$, $\hat{\mu}_{A}(\mathbf{X})\overset{\mathbb{P}}{\to}{\mu}_{A}(X)$, and $\hat{\mu}_{A}(\mathbf{Y})\overset{\mathbb{P}}{\to}{\mu}_{A}(Y)$ as $ n\to\infty$. For brevity, we only show the first convergence
         \[
         \frac{1}{n}\sum_{i=1}^n X_i Y_i 1_{\hat{A}}(X_i,Y_i)\overset{\mathbb{P}}{\to} \mathbb{E}[XY 1_A(X,Y)],\quad n\to\infty;
        \]
        the remaining arguments are similar. To do so, we first note that
        \begin{align}\label{eq:th:consistency:1}
            \frac{1}{n}\sum_{i=1}^n X_i Y_i 1_{\hat{A}}(X_i,Y_i) = \frac{1}{n}\sum_{i=1}^n X_i Y_i 1_{A}(X_i,Y_i) + \frac{1}{n}\sum_{i=1}^n X_i Y_i \left(1_{\hat{A}}(X_i,Y_i)-1_{A}(X_i,Y_i)\right).
        \end{align}
        Also, as $(X_i Y_i 1_{A}(X_i,Y_i))_{i=1}^\infty$ is an i.i.d. sequence, by the law of large numbers, we get
        \[
        \frac{1}{n}\sum_{i=1}^n X_i Y_i 1_{A}(X_i,Y_i)\overset{\mathbb{P}}{\to} \mathbb{E}[XY 1_A(X,Y)],\quad n\to\infty.
        \]
        Thus, to conclude the proof, it is enough to show $\frac{1}{n}\sum_{i=1}^n X_i Y_i \left(1_{\hat{A}}(X_i,Y_i)-1_{A}(X_i,Y_i)\right)\overset{\mathbb{P}}{\to} 0$ as $n\to\infty$. Let $A_X:=[Q_X(p_1),Q_X(q_1)]$, $A_Y:=[Q_Y(p_2),Q_Y(q_2)]$, $\hat A_X:=[X_{([np_1]+1)},X_{([nq_1])}]$, and $\hat A_X:=[Y_{([np_2]+1)},Y_{([nq_2])}]$. Then, we get
\begin{align}\label{eq:lm:conv_normalisation:1}
      \left|\sum_{i=1}^n X_i Y_i\left(I_A(X_i,Y_i) - I_{\hat{A}}(X_i,Y_i)\right)\right| & \leq 
      \sum_{i=1}^n M_n \left(1_{\{X_i \in A_X, X_i\notin \hat A_X\}}+1_{\{X_i \notin A_X, X_i\in \hat A_X\}}+1_{\{Y_i \in A_Y, Y_i\notin \hat A_Y\}}+1_{\{Y_i \notin A_Y, Y_i\in \hat A_Y\}}\right)\nonumber\\
      &\leq M_n\left(2|p_1^n -[np_1 ]|+2|q_1^n -[nq_1 ]| + 2|p_2^n -[np_2 ]|+2|[nq_2] - q_2^n |\right).
\end{align}
where $M_n:= M_n^X  M_n^Y$, and the supplementary notation
\[
p^n_1:= \sum_{i=1}^n 1_{\{X_i \leq Q_X(p_1)\}}, \quad
q^n_1:= \sum_{i=1}^n 1_{\{X_i \leq Q_X(q_1)\}}, \quad
p^n_2 :=\sum_{i=1}^n 1_{\{Y_i \leq Q_Y(p_2)\}}, \quad
q^n_2 :=\sum_{i=1}^n 1_{\{Y_i \leq Q_Y(q_2)\}},
\]
\[
M_n^X:=\max(|X_{([np_1])}|,|X_{([nq_1])}|,|X_{(p_1^n)}|,|X_{(q_1^n)}|), \quad
M_n^Y:=\max(|Y_{([np_2])}|,|Y_{([nq_2])}|,|Y_{(p_2^n)}|,|Y_{(q_2^n)}|),
\]
for $n\in\bN$ is used. Note that, by the law of large numbers, we get
\begin{equation*}\label{eq:lm:conv_normalisation:2}
    \tfrac{1}{n}p_1^n\overset{\mathbb{P}}{\to}p_1, \quad \tfrac{1}{n}q_1^n\overset{\mathbb{P}}{\to}q_1, \quad \tfrac{1}{n}p_2^n\overset{\mathbb{P}}{\to}p_2, \quad\tfrac{1}{n}q_2^n\overset{\mathbb{P}}{\to}q_2, \quad \text{as} \quad n\to\infty.
\end{equation*}
Also, we get
\begin{equation*}\label{eq:lm:conv_normalisation:3}
\tfrac{1}{n}[np_1 ]\to p_1, \quad \tfrac{1}{n}[nq_1 ]\to q_1, \quad \tfrac{1}{n}[np_2 ]\to p_2, \quad \tfrac{1}{n}[nq_2 ]\to q_2, \quad \text{as} \quad n\to\infty.
\end{equation*}
Finally, by the consistency of the quantile estimators, we get $M_n\overset{\mathbb{P}}{\to} \max(|Q_X(p_1)|,|Q_X(q_1)|)\cdot \max(|Q_Y(p_2)|,|Q_Y(q_2)|)$.
Consequently, dividing~\eqref{eq:lm:conv_normalisation:1} by $n$ and letting $n\to\infty$, we get 
\begin{align}\label{eq:th:conv_normalisation:4}
    \frac{1}{n}\sum_{i=1}^n X_i Y_i \left(1_{\hat{A}}(X_i,Y_i)-1_{A}(X_i,Y_i)\right)\overset{\mathbb{P}}{\to} 0, \quad n\to\infty,
\end{align}
which concludes the proof.
\end{proof}
Next, let us show that, when the set $A$ is known, we get the asymptotic normality of the conditional correlation estimator. 

\begin{theorem}\label{th:normality_known_set}
    Let $\cor_A(X,Y)$ and $\overline\cor_A(\mathbf{X},\mathbf{Y})$ be given by~\eqref{eq:cond_cov} and~\eqref{eq:cond_cov_emp_known}, respectively. Then, for some $\tau>0$, we get
    \[
    \sqrt{n}\left(\overline\cor_A(\mathbf{X},\mathbf{Y})-\cor_A(X,Y)\right)\overset{d}{\longrightarrow} \mathcal{N}(0,\tau), \quad n\to \infty.
    \]
\end{theorem}
\begin{proof}
First note that, by the analogy to~\eqref{eq:th:consistency:0}, we have
\begin{align*}
            \overline\cov_A(\mathbf{X},\mathbf{Y}) = \frac{n}{\sum_{i=1}^n 1_{{A}}(X_i,Y_i)} \frac{1}{n}\sum_{i=1}^n X_i Y_i 1_{{A}}(X_i,Y_i)- \bar{\mu}_{A}(\mathbf{X})\bar{\mu}_{A}(\mathbf{X}). 
\end{align*}
Setting $W_i:=[1_{{A}}(X_i,Y_i), X_i 1_{{A}}(X_i,Y_i), Y_i 1_{{A}}(X_i,Y_i), X_i Y_i 1_{{A}}(X_i,Y_i), X_i^21_{{A}}(X_i,Y_i), Y_i^21_{{A}}(X_i,Y_i)]^{T}$, for $i\in \mathbb{N}\setminus \{0\}$, we get
\[
\overline{\cor}_A(\mathbf{X},\mathbf{Y})=g\left(\frac{1}{n}\sum_{i=1}^n W_i\right),
\]
where $g(t_1,t_2,t_3,t_4,t_5,t_6):=\frac{t_4/t_1-(t_2/t_1)(t_3/t_1)}{\sqrt{t_5/t_1-(t_2/t_1)^2}\sqrt{t_6/t_1-(t_3/t_1)^2}}$, for $t_1>0$, $t_2,t_3,t_4, t_5, t_6\in \mathbb{R}$. Next, noting that $(W_i)$ is an i.i.d. sequence of square integrable random vectors, by the multidimensional central limit theorem, we get
    \[
    \sqrt{n}\left( \frac{1}{n}  \sum_{i=1}^n W_i -\Theta\right)  \overset{d}{\to} \mathcal{N}_6(0,\Sigma), \quad n\to \infty,
    \]
    where $\Theta:=[
        \mathbb{P}[(X,Y)\in A], \mathbb{E}[X 1_{\{(X,Y)\in A\}}], \mathbb{E}[Y 1_{\{(X,Y)\in A\}}], \mathbb{E}[XY 1_{\{(X,Y)\in A\}}],\mathbb{E}[X^2 1_{\{(X,Y)\in A\}}],\mathbb{E}[Y^2 1_{\{(X,Y)\in A\}}]]^T$ and $\Sigma$ is some positive-definite ($6\times 6$) covariance matrix. Consequently, using the delta method (see e.g. Theorem 7 in~\cite{Ferguson1996}), we get
    \begin{align}\label{eq:th:norm_known_set:5}
\sqrt{n}\left(\overline\cor_A(X,Y)-\cor_A(X,Y)\right) = \sqrt{n}\left(g\left(\frac{1}{n}  \sum_{i=1}^n W_i\right)-g(\Theta)\right)\overset{d}{\to} \mathcal{N}(0,\tau), \quad n\to \infty,
    \end{align}
    for $\tau:= \nabla g(\Theta))^T \Sigma \nabla g(\Theta)$,
    which concludes the proof.
\end{proof}

It should be noted that Theorem~\ref{th:normality_known_set} shows the asymptotic normality (and, \textit{a fortiori}, the consistency) of the conditional sample covariance when the conditioning is based on a fixed set $A$ instead of a sample-dependent set $\hat{A}$.  While one expects that similar results hold also for $\hat{A}$ and  $\widehat\cov_A(\mathbf{X},\mathbf{Y})$, we were not able to show this result using the standard delta-method together with multidimensional Central Limit Theorem. That said, one does not expect that the estimation risk embedded in the estimation of set $A$ plays a dominant role in covariance estimation risk. Let us now show two simple examples that illustrate the consistency property together with the rate of the corresponding convergence, as well as estimation risks.

\begin{example}[Bivariate normal vector] \label{ex:MVN}
Let $(X,Y)\sim \mathcal{N}_2( \mu, \Sigma)$, where $\mu=(0.5,0.5)^T$ and $\Sigma$ is such that $\sigma^2(X)=\sigma^2(Y)=1$ and $\cov(X,Y)=0.4$. Let $p_1=p_2=\Phi(0.05)$ and $q_1=q_2=\Phi(0.75)$ so that we get $A = [0.05,0.75] \times [0.05,0.75]$. Using simple Monte Carlo estimation, we can approximate conditional moments of $(X,Y)$. In particular, for an exemplary Monte Carlo sample of size $N=1000$, we get
\[
\mu_A(X)  = \mu_A(Y) \approx 0.16,\quad \sigma_A^2(X)=\sigma_A^2(Y) \approx 0.37,\quad \cov_A(X,Y) \approx 0.06,
\]
from which we can deduce that $\cor_A(X,Y)\approx 0.16$. This shows that the conditional correlation could be lower than the unconditional one. In Figure~\ref{MCplotMultiNormal} we illustrate the consistency of the sample correlation estimator as well as its mean squared error. We can see that the sample estimate is close to the theoretical value (left panel) and the mean squared error converges to $0$ (middle panel). Also, we can see that the mean squared difference between  $\widehat\cor_A(X,Y)$ and $\overline\cor_A(X,Y)$ also converges to $0$ (right panel) which suggests that the statement of Theorem~\ref{th:normality_known_set} could be extended to the estimator with the random set $\hat{A}$. In fact, as expected, the simulation results confirm that the estimation error linked to $\hat{A}$-set approximation procedure plays minor role in the total simulation error. More specifically, the ratio of the error resulting from the estimation of $\hat{A}$ to the total error of  $\widehat{\cor}_A(X,Y)$ is $0.3$ for $n=200$ and decreases to $0.1$ for $n=2000$.

\begin{figure}[htp!]
\begin{center}
\includegraphics[width=0.32\textwidth]{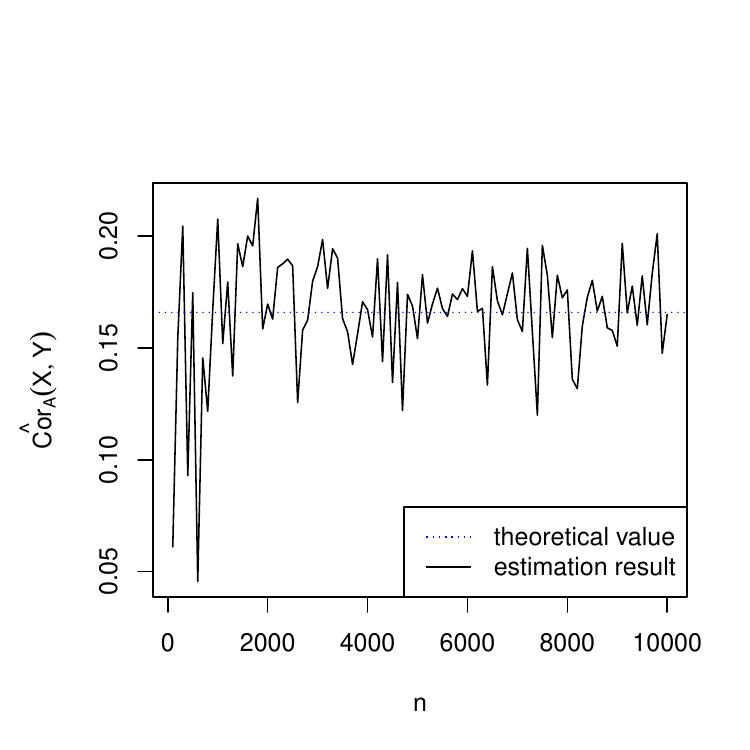}
\includegraphics[width=0.32\textwidth]{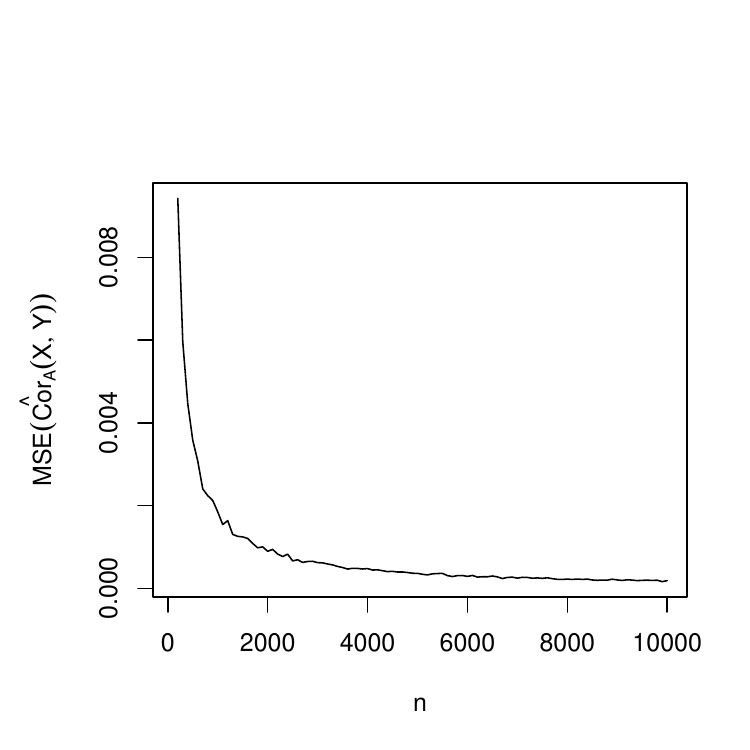}
\includegraphics[width=0.32\textwidth]{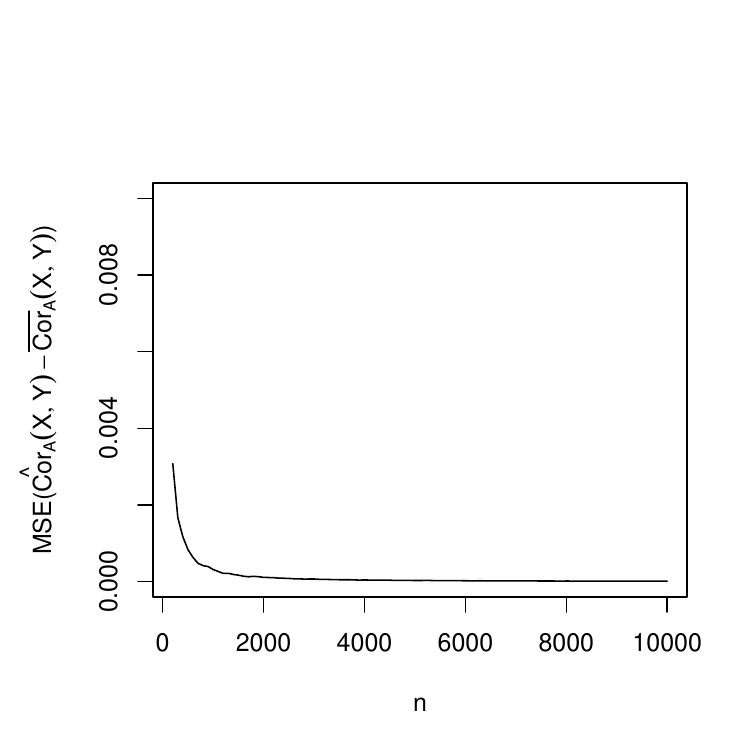}
\end{center}
\caption{The left panel presents a convergence of the estimated conditional correlation on a single trajectory (with the blue dotted line corresponding to the theoretical value), the middle panel shows the mean squared error (MSE) while the right panel shows the mean squared difference between $\widehat\cor_A(X,Y)$ and $\overline\cor_A(X,Y)$ for $A = [0.05,0.75]^2$. The simulations were based on $1000$ trajectories for each $n$ and the bivariate Gaussian random vector $(X,Y)$ discussed in Example~\ref{ex:MVN}. }\label{MCplotMultiNormal}
\end{figure}
\end{example}

\begin{example}[Bivariate $\alpha$-stable vector]\label{ex:MVS}
   Let $(X,Y)$ be a bivariate symmetric $\alpha$-stable distribution with the shape parameter $\alpha=1.5$, the spectral measure   
\begin{eqnarray}
 \Gamma=\frac{1}{4}\delta\left(\frac{\sqrt{2}}{2},\frac{\sqrt{2}}{2}\right) + \frac{1}{4}\delta\left(-\frac{\sqrt{2}}{2},-\frac{\sqrt{2}}{2}\right) +\frac{1}{4}\delta\left(-\frac{\sqrt{2}}{2},\frac{\sqrt{2}}{2}\right)+\frac{1}{4}\delta\left(\frac{\sqrt{2}}{2},-\frac{\sqrt{2}}{2}\right),
\end{eqnarray}
and the characteristic function
\begin{eqnarray}
\Phi(\theta):=\mathbb{E}\left(\exp\{i\theta_1X+i\theta_2Y\}\right)=\exp\left\{-\int_{S_2}|(\theta,s)|^{\alpha}\Gamma(ds_1,ds_2)\right\},\quad \theta=(\theta_1,\theta_2)\in \mathbb{R}^2,
\end{eqnarray}
where $\delta(a,b)$ is the Dirac delta measure at $(a,b)\in \mathbb{R}^2$ and $S_2$ is the unit sphere, see~\cite{SamTaq1994} for details. The random variables $X$ and $Y$ are univariate symmetric $\alpha-$stable distributed with the same distribution functions characterized by the stability parameter $\alpha$ and scale parameters  $\sigma_X^{\alpha}=\sigma_Y^{\alpha}=\int_{S_2}|s_1|^{\alpha}\Gamma(ds_1,ds_2)$; note that the unconditioned second moments are infinite. As in Example~\ref{ex:MVN}, let us fix $p_1=p_2=F_X(0.05)$ and $q_1=q_2=F_X(0.75)$ so that we get $A = [0.05,0.75] \times [0.05,0.75]$. As in the previous example, in Figure \ref{Plot:stable:1_5} we illustrate the consistency of the sample correlation estimator as well as its mean squared error. Even though the underlying distribution is not square-integrable, we can see that the conditional correlation estimator is consistent. As in Example~\ref{ex:MVN}, we can see that the mean squared error converges to $0$ as $n\to\infty$. Also, the ratio of the $\hat{A}$-set estimation error to the total estimation error is  $0.33$ for $n=200$ and $0.15$ for $n=2000$. This again suggests that the statement of Theorem~\ref{th:normality_known_set} could be extended to the estimated set $A$.
\begin{figure}[htp!]
\begin{center}
\includegraphics[width=0.32\textwidth]{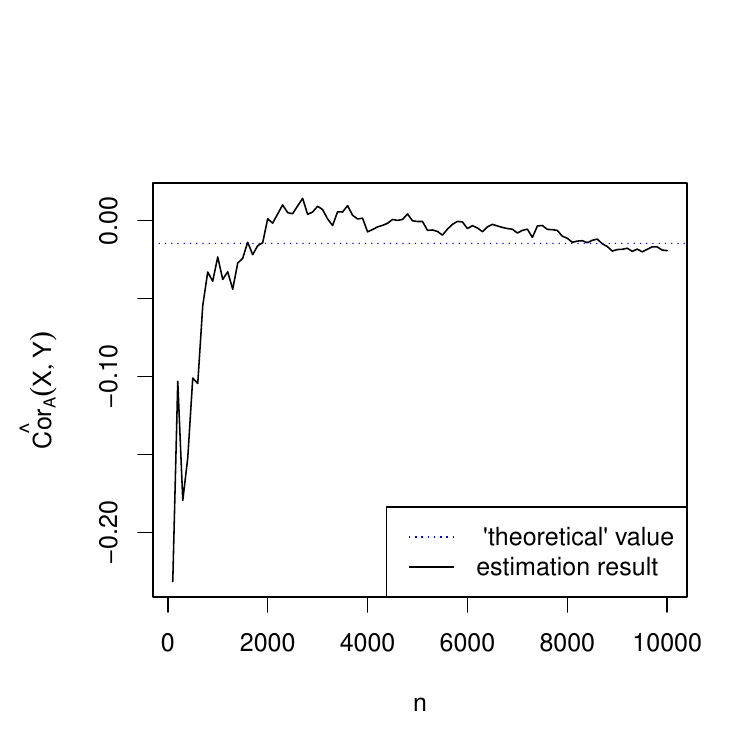}
\includegraphics[width=0.32\textwidth]{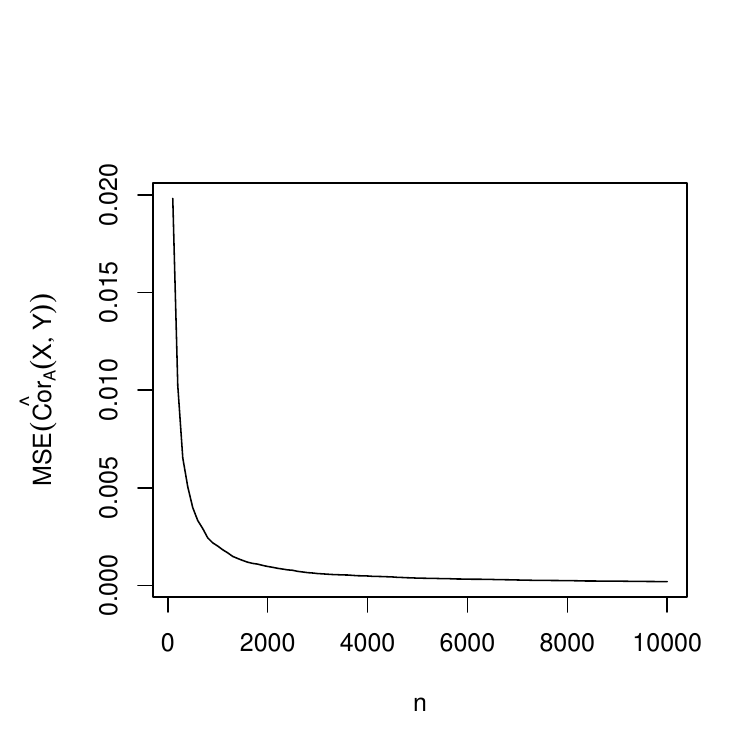}
\includegraphics[width=0.32\textwidth]{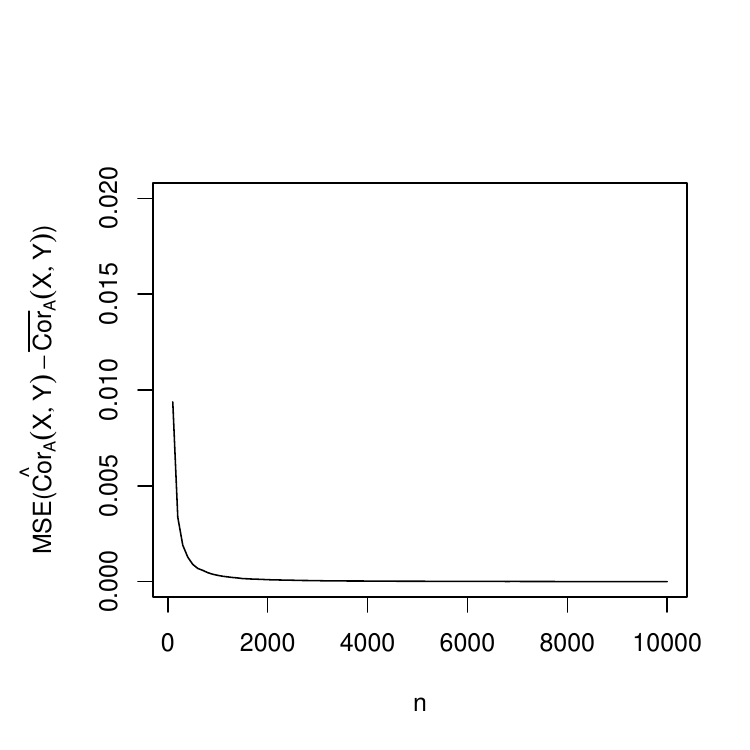}
\end{center}
\caption{
The left panel presents a convergence of the estimated conditional correlation on a single trajectory (with the blue dotted line corresponding to the value obtained on 50000-element sample), the middle panel shows the mean squared error (MSE) while the right panel shows the mean squared difference between $\widehat\cor_A(X,Y)$ and $\overline\cor_A(X,Y)$ for $A = [0.05,0.75]^2$. The simulations were based on $1000$ trajectories for each $n$ and the bivariate Gaussian random vector $(X,Y)$ discussed in Example~\ref{ex:MVS}.}\label{Plot:stable:1_5}
\end{figure}
\end{example}


\section{Conditional autocorrelation function and its estimator}\label{S:cond_autocor}

In this section, we show how to translate the methods introduced in Section~\ref{S:preliminaries} and Section~\ref{S:cond_cor_est} to the time series framework. We consider the serial dependence detection problem and introduce a  new tools which facilitate statistical testing of independence of time series observations. To do so, we define a conditional autocorrelation measure for a time series, discuss its estimator, and show some basic properties.

In the following, we assume that we are given a real valued time series $\{X_t\}_{t \in \mathbb{Z}}$ that is strictly stationary with one-dimensional cumulative distribution function $F$; we assume that $F$ is invertible with the inverse $Q$.
In this section, since we are considering a stationary time series, we simplify the definition of the set A from~\eqref{eq:set_A}. More specifically,  for any fixed quantile split $0<p<q<1$, with the slight abuse of notation, we define
\begin{equation}\label{eq:set_A_autocor}
    A:=[Q(p),Q(q)]\times [Q(p),Q(q)].
\end{equation}
Note that this effectively induces the same conditioning for both coordinates. Next, for an arbitrary $t\in \mathbb{Z}$ and $h\in \mathbb{N}$, we define the lag $h$ conditional autocorrelation by setting
\begin{equation}\label{eq:autocov_unit}
    \rho_A(h):=\cor_A(X_t,X_{t+h});
\end{equation}
directly from the stationarity of $\{X_t\}_{t \in \mathbb{Z}}$, we conclude that~\eqref{eq:autocov_unit} does not depend on $t\in \mathbb{Z}$. 
In the following, for the notational convenience, if we want to emphasise the dependence of $\rho_{A}(h)$ on the quantile split $0<p<q<1$, we  write
$\rho_{(p,q)}(h)$ for $\rho_{A}(h)$ corresponding to the conditioning set $A:=[Q(p),Q(q)]\times [Q(p),Q(q)]$.
Also, by analogy to Section~\ref{S:preliminaries},  we use $\rho(h)$ to denote the unconditional lag $h$ autocorrelation; this effectively corresponds to setting $p=0$ and $q=1$. We refer to the map $h\to \rho_A(h)$ as conditional autocorrelation fuction (CACF) and to its unconditional version $h\to \rho(h)$ simply as autocorrelation function (ACF).

By analogy to unconditional autocorrelation, the condition autocorrelation measures a linear dependence between the observations in a time series with indices lagged by $h$. However, the conditional version is well-defined even if the random variables are not square-integrable. Also, directly from Theorem~\ref{th:independence}, we get that if $\{X_t\}_{t \in \mathbb{Z}}$ are independent, then $\rho_A(h)=0$ for any choice of $A$. In fact, by using all positive lags and multidimensional conditioning, one could obtain a complete characterisation of independence; 
see Theorem 3.2 and Remark 5 from~\cite{JawJelPit2022} 
for details. 



Next, given $h\in\bN$, we discuss how to estimate~\eqref{eq:autocov_unit} using a {finite sample $\{X_t\}_{t=1}^{n}$}, where $n\in \mathbb{N}\setminus \{0\}$. To ease the notation, for $k=1, \ldots, n$, let $r_1(k)$ denote $k$th order statistic in the sample $\{X_t\}_{t=1}^{n}$, i.e. index corresponding to the $k$th smallest value. Similarly,  for $k=1, \ldots, n$, let $r_2(k)$ denote $k$th order statistic in the sample $\{X_t\}_{t=1+h}^{n+h}$; note that we assume {$r_2$} takes values in the set $\{1+h, \ldots, n+h\}$. Then, the sample version of the set $A$ from~\eqref{eq:set_A_autocor} can be defined by setting
\begin{equation}\label{eq:set_A.hat.auto}
\hat{A} :=[X_{r_1([np]+1)},X_{r_1([nq])}] \times [X_{r_2([np]+1)},X_{r_2([nq])}].
\end{equation}
Combining this definition with~\eqref{eq:cond_cor_est}, we define the conditional autocorrelation estimator by
\begin{equation}
     \hat\rho_A(h):=\widehat\cor_A((X_t)_{t=1}^n,(X_t)_{t=1+h}^{n+h})\label{eq:autocor_unit_est}
\end{equation}
together with the supplementary estimator
\begin{equation}
\bar\rho_A(h):=\overline\cor_A((X_t)_{t=1}^n,(X_t)_{t=1+h}^{n+h}),\label{eq:autocor_unit_est_known_set}
\end{equation}
in which the set $A$ is known and need not be estimated. As before, we refer to the map $h\to \hat\rho_A(h)$ as sample conditional autocorrelation function (sample CACF) and to its unconditional version $h\to \hat\rho(h)$ as sample autocorrelation function (sample ACF). Sometimes, with a slight abuse of notation, we drop the word {\it sample} and call $h\to \hat\rho_A(h)$ and $h\to \hat\rho(h)$ simply CACF and ACF, respectively.

Now, we translate the properties of  $\widehat\cor_A$ and $\overline\cor_A$ stated in Theorem~\ref{th:consistency} and Theorem~\ref{th:normality_known_set} to the time series framework.
It should be noted that in the following results' statements, we assume that the time series constituents $\{X_t\}_{t \in \mathbb{Z}}$ are independent, which corresponds to independence (null) hypothesis in testing for which test statistic distribution is typically derived. First, we show the consistency of the conditional autocorrelation.

\begin{theorem}\label{th:consistency_ts}
   For any $h\in \mathbb{N}$, let $\rho_A(h)$ and $\hat\rho_A(h)$ be given by~\eqref{eq:autocov_unit} and~\eqref{eq:autocor_unit_est}, respectively. Also, let $\{X_t\}_{t \in \mathbb{Z}}$ be an i.i.d. time series. Then, we get
    \[
    \hat\rho_A(h)\overset{\mathbb{P}}{\longrightarrow} \rho_A(h), \quad n\to \infty.
    \]
\end{theorem}
\begin{proof}
The proof will be shown only for $h=1$; the general case follows the same logic. Also, as in Theorem~\ref{th:consistency} we only prove
\[
    \frac{1}{n}\sum_{i=1}^n X_i X_{i+1} 1_{\hat{A}}(X_i,X_{i+1})\overset{\mathbb{P}}{\to} \mathbb{E}[X_iX_{i+1} 1_A(X_i,X_{i+1})],\quad n\to\infty;
\]
using a similar argument one can show the consistency of other empirical moments. In the following, to ease the notation, we set $A_1:=A_2:=[Q(p),Q(q)]$, $\hat{A}_1:=[X_{r_1([np]+1)},X_{r_1([nq])}]$, and $\hat{A}_2:=[X_{r_2([np]+1)},X_{r_2([nq])}]$. Next, note that, as in~\eqref{eq:th:consistency:1}, we have
\begin{align*}
    \frac{1}{n}\sum_{i=1}^n X_i X_{i+1} 1_{\hat{A}}(X_i,X_{i+1})= \frac{1}{n}\sum_{i=1}^n X_i X_{i+1} 1_{A}(X_i,X_{i+1}) + \frac{1}{n}\sum_{i=1}^n X_iX_{i+1} \left( 1_{\hat{A}_1}(X_i) 1_{\hat{A}_2}(X_{i+1}) - 1_{A_1}(X_i)1_{A_2}(X_{i+1})
    \right).
\end{align*}
Now, let us show that the first term converges to $\mathbb{E}[X_iX_{i+1} 1_A(X_i,X_{i+1})]$ while the second term converges to zero. First, note that the law of large numbers is not directly applicable to $\{X_i X_{i+1} 1_{A}(X_i,X_{i+1})\}_{i=1}^\infty$ as this sequence may not be independent. However, setting $Z_i := X_iX_{i+1}1_{A_1}(X_i)1_{A_2}(X_{i+1})$, $i\in \mathbb{N}\setminus \{0\}$ and $\lceil x\rceil:=\min\{n\in \mathbb{Z}\colon n\geq x\}$, $x\in \mathbb{R}$, we obtain
\begin{align}
\frac{1}{n}\sum_{i=1}^n X_i X_{i+1} 1_{A}(X_i,X_{i+1}) =\frac{1}{2}\left(\frac{2}{n} \sum_{i=1}^{[\frac{n}{2}]}Z_{2i} + \frac{2}{n} \sum_{i=1}^{\lceil\frac{n}{2}\rceil}Z_{2i-1} \right), \quad n\in \mathbb{N}\setminus \{0\},
\end{align}
and the sequences $\{Z_{2i}\}_{i=1}^\infty$ and $\{Z_{2i-1}\}_{i=1}^\infty$ are both i.i.d. Thus, using the law of large numbers twice, we get
\begin{align}
    \frac{1}{n}\sum_{i=1}^n X_i X_{i+1} 1_{A}(X_i,X_{i+1}) \overset{\mathbb{P}}{\to} \frac{1}{2} \mathbb{E}[Z_2] + \frac{1}{2} \mathbb{E}[Z_1] = \mathbb{E}\left[X_i X_{i+1} 1_{A}(X_i,X_{i+1}) \right], \quad n \to \infty.
\end{align}
Second, following the logic leading to~\eqref{eq:th:conv_normalisation:4}, we get
\[
\frac{1}{n}\sum_{i=1}^n X_iX_{i+1} \left( 1_{\hat{A}_1}(X_i) 1_{\hat{A}_2}(X_{i+1}) - 1_{A_1}(X_i)1_{A_2}(X_{i+1})
    \right)\overset{\mathbb{P}}{\to} 0, \quad n\to\infty, 
\]
which concludes the proof.
\end{proof}
Second, by analogy to Theorem~\ref{th:normality_known_set}, we show that the conditional autocorrelation with known set $A$ is asymptotically normal. 

\begin{theorem}\label{th:th:normality_known_set_ts}
   For any $h\in \mathbb{N}$, let $\rho_A(h)$ and $\bar\rho_A(h)$ be given by~\eqref{eq:autocov_unit} and~\eqref{eq:autocor_unit_est_known_set}, respectively. Also, let $\{X_t\}_{t \in \mathbb{Z}}$ be an i.i.d. time series. Then, for some $\tau>0$, we get
    \[
    \sqrt{n}\left(\bar\rho_A(1)-\rho_A(h)\right)\overset{d}{\longrightarrow} \mathcal{N}(0,\tau), \quad n\to \infty.
    \]
\end{theorem}
\begin{proof}
As in the proof of Theorem~\ref{th:normality_known_set}, we start with noting that 
\[
\bar\rho_A(1)=g\left(\frac{1}{n}\sum_{i=1}^n W_i\right),
\]
where $g(t_1,t_2,t_3,t_4,t_5,t_6):=\frac{t_4/t_1-(t_2/t_1)(t_3/t_1)}{\sqrt{t_5/t_1-(t_2/t_1)^2}\sqrt{t_6/t_1-(t_3/t_1)^2}}$, $t_1\in \mathbb{R}>0$, $t_2,t_3,t_4, t_5, t_6\in \mathbb{R}$, and for $i\in \mathbb{N}\setminus \{0\}$ we have
\begin{align*}
    W_i&:=\begin{bmatrix}
        1_{{A}}(X_i,X_{i+1}), X_i 1_{{A}}(X_i,X_{i+1}), X_{i+1} 1_{{A}}(X_i,X_{i+1}), X_i X_{i+1} 1_{{A}}(X_i,X_{i+1}), X_i^21_{{A}}(X_i,X_{i+1}), X_{i+1}^21_{{A}}(X_i,X_{i+1})
    \end{bmatrix}^{T}.
\end{align*}
Next, note that $\{W_i\}_{i=1}^\infty$ is 1-dependent in the sense of Definition 6.4.3 in \cite{BroDav1991}. Thus, combining Theorem 6.4.2 therein with Cram\'{e}r-Wald device (see e.g. Theorem 29.4 in~\cite{Bil1995}),
we obtain
    \[
    \sqrt{n}\left( \frac{1}{n}  \sum_{i=1}^n W_i -\Theta\right)  \overset{d}{\to} \mathcal{N}_6(0,\Sigma), \quad n\to \infty,
    \]
    where $\Theta:=[
        \mathbb{P}[(X,Y)\in A], \mathbb{E}[X 1_{\{(X,Y)\in A\}}], \mathbb{E}[Y 1_{\{(X,Y)\in A\}}], \mathbb{E}[XY 1_{\{(X,Y)\in A\}}],\mathbb{E}[X^2 1_{\{(X,Y)\in A\}}],\mathbb{E}[Y^2 1_{\{(X,Y)\in A\}}]
    ]^T$
    and $\Sigma$ is some positive definite matrix. Thus, repeating the argument leading to~\eqref{eq:th:norm_known_set:5}, we conclude the proof.
\end{proof}
Theorem~\ref{th:consistency_ts} and Theorem~\ref{th:th:normality_known_set_ts} could be seen as the conditional versions of the standard results on the consistency and asymptotic normality of the unconditional autocorrelation; see e.g. Chapter 7 in~\cite{BroDav1991} for details. Note that, however, the classical results for the unconditional autocorrelation typically require square integrability of time series elements; see also~\cite{DavRes1986} for the autocorrelation convergence for non-integrable stable distributions. Due to the conditioning, in Theorem~\ref{th:consistency_ts} and Theorem~\ref{th:th:normality_known_set_ts} we do not need any integrability assumption.

The classical unconditional autocorrelation measure $\rho(h)$ is often used to construct the ACF, which links the lag $h$ with the (estimated) values of $\rho(h)$. This tool finds numerous applications in time series analysis; see e.g. Chapter 3 in~\cite{brockwell2016introduction} for the connection with the linear model's order identification. As already stated, one can also construct a conditional autocorrelation function, i.e. the map $h\mapsto \rho_A(h)$. In fact, in Example~\ref{ex:ACF} we discuss simple usage of its sample version $h\mapsto \hat\rho_A(h)$. Based on the example, one can see that the conditional ACF preserves the basic features of ACF while it is also able to provide some additional insight on the data correlation structure. Furthermore, by analogy to the conditional autocorrelation function, one can define a conditional partial autocorrelation function; see e.g. Section 3.2.3 in~\cite{brockwell2016introduction} for the discussion on its unconditional version. However, due to the technical complexity related to the two-level conditioning, this is left for future research.

\begin{example}[ACF and CACF for AR(1) and GARCH(1,1) processes]\label{ex:ACF}
    Let $\{Z_t\}_{t \in \mathbb{Z}}$ be an autoregressive process of order 1 (AR(1)) given by $Z_t := 0.5 Z_{t-1} + \varepsilon_t$, $t\in \mathbb{Z}$, where $\varepsilon_t$ is i.i.d. $\mathcal{N}(0,1)$. In Figure~\ref{plot:acfAR} we show the autocorrelation function and the conditional autocorrelation function for $p = 0.01$ and $q = 0.99$, estimated on a sample consisting of 250 elements. As expected, the conditional autocorrelation exhibits behaviour similar to that of its unconditional version. Now, let $\{Z_t\}_{t \in \mathbb{Z}}$ be GARCH(1,1) process with $\omega_0 = 0.01$, $\omega_1 = 0.6$ and $\omega_2 = 0.2$; see Equation~\eqref{eq:garch_model} for the explicit definition. In Figure~\ref{plot:acfGARCH} we show the autocorrelation function and the conditional autocorrelation function (with $p = 0.01$ and $q=0.65$) for 250-element sample from this model. It can be seen that, in contrast to the standard autocorrelation function, the conditional autocorrelation function can detect dependencies in the GARCH time series; see Section~\ref{S:GARCH} for a more detailed analysis.

    \begin{figure}[htp!]
\begin{center}
\includegraphics[width=0.33\textwidth]{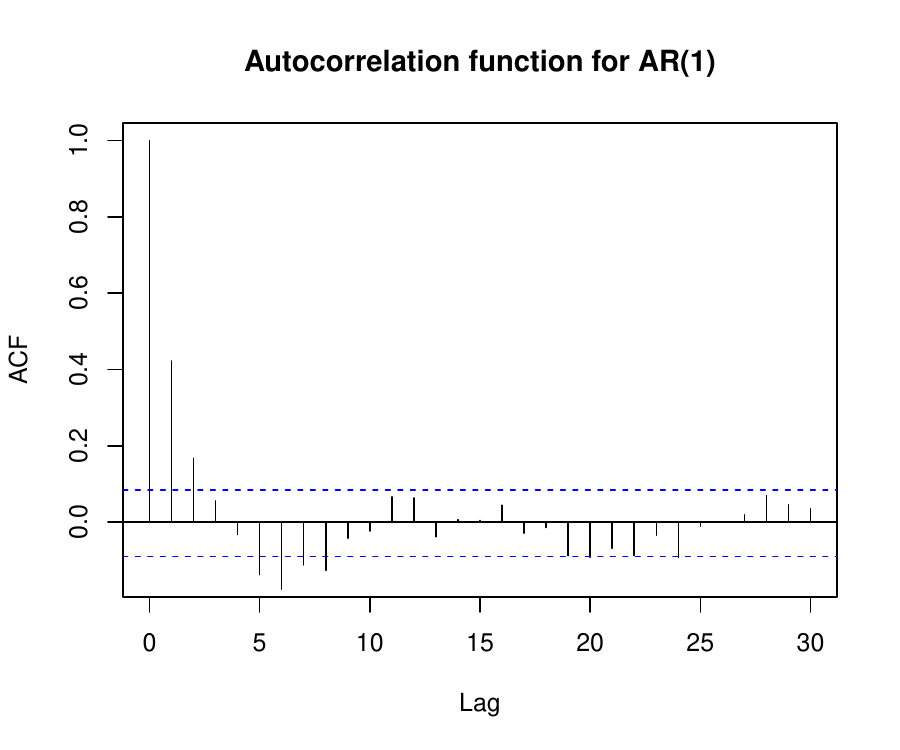}
\includegraphics[width=0.33\textwidth]{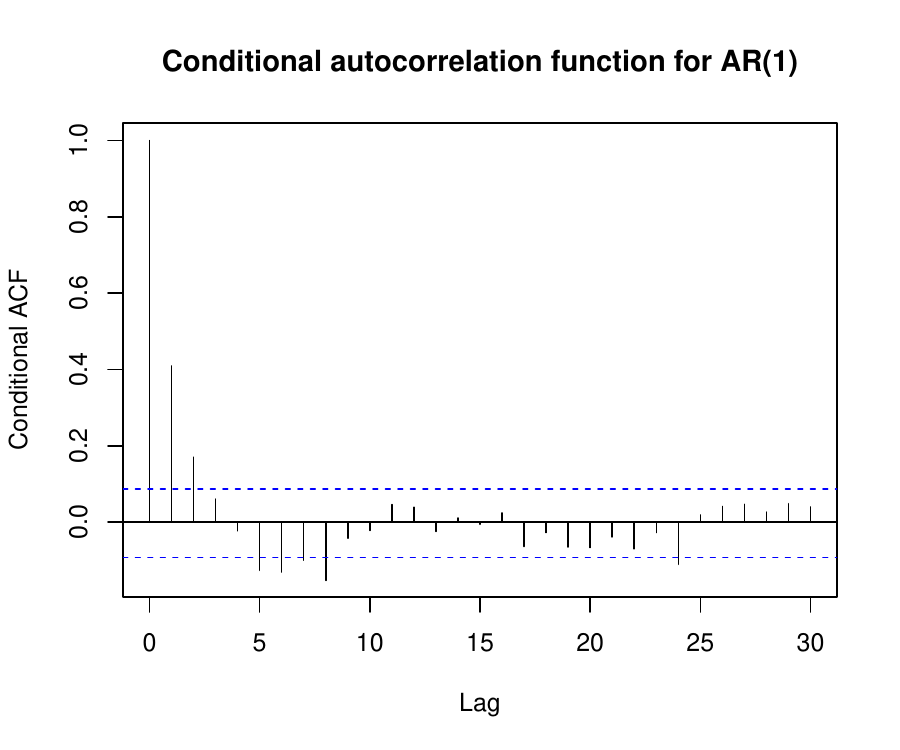} 
\end{center}
\caption{The conditional autocorrelation and the autocorrelation function for $AR(1)$ process witch $\varphi = 0.5$, and quantile splits $p = 0.01$, $q =0.99$. Blue dashed lines represent confidence intervals calculated based on Monte Carlo simulations i.e. $95\%$ confidence interval of the simulated results under the null hypothesis that the sample is from Gaussian white noise.}\label{plot:acfAR}
\end{figure}

\begin{figure}[htp!]
\begin{center}
\includegraphics[width=0.33\textwidth]{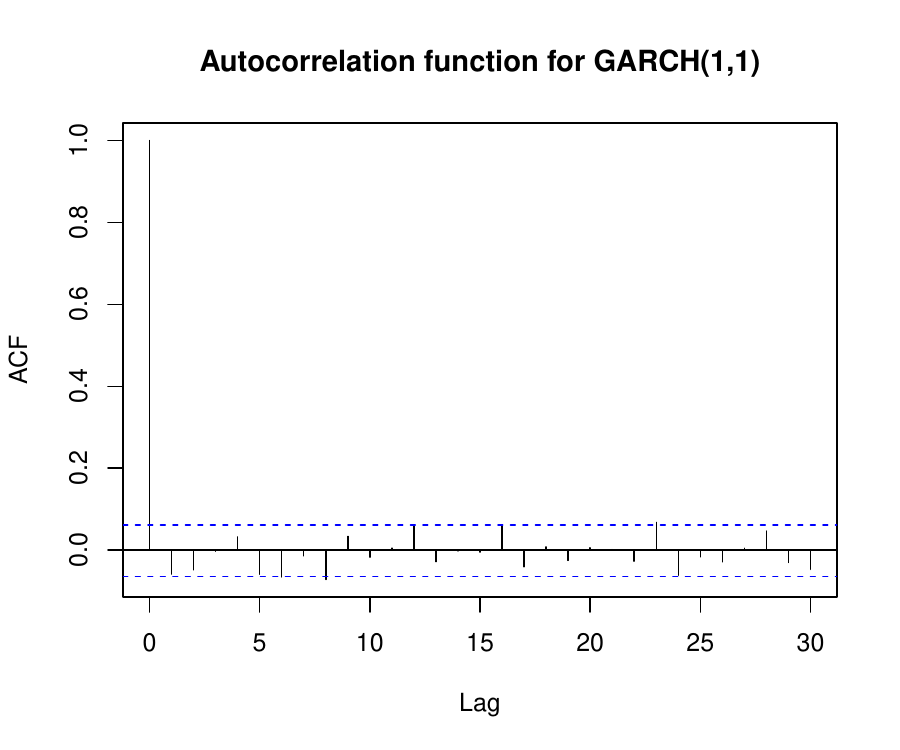}
\includegraphics[width=0.33\textwidth]{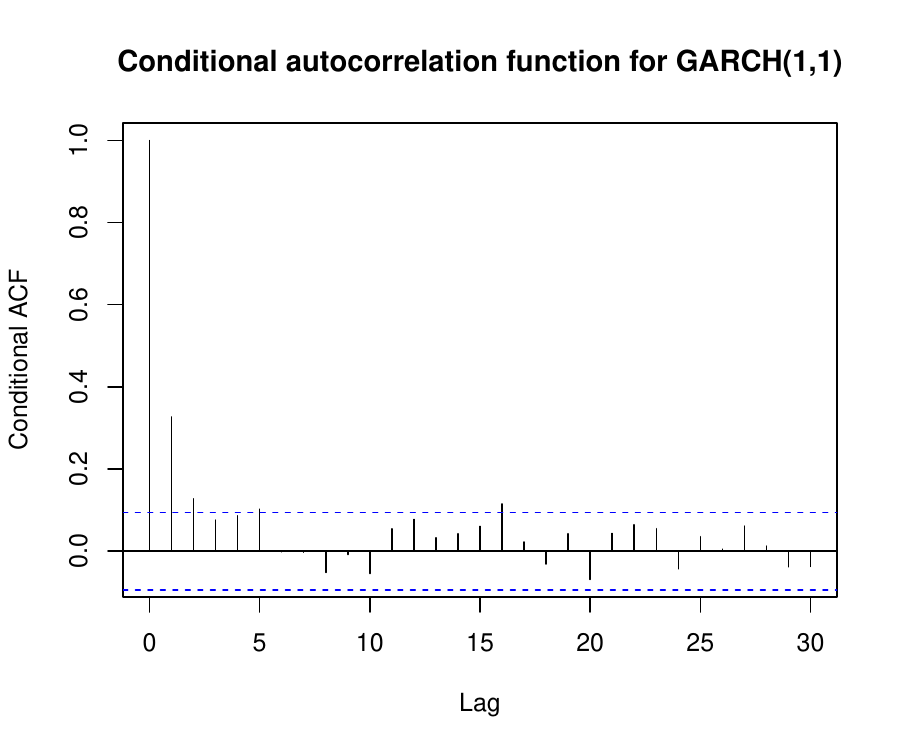}
\end{center}
\caption{The conditional autocorrelation and the autocorrelation function for $GARCH(1,1)$ process witch $\omega_1 = 0.01$, $\omega_2 = 0.6$, and $\omega_2 = 0.2$, and quantile splits $p = 0.01$, $q =0.65$. Blue dashed lines represent confidence intervals calculated based on Monte Carlo simulations i.e. $95\%$ confidence interval of the simulated results under the null hypothesis that the sample is from Gaussian white noise.}\label{plot:acfGARCH}
\end{figure}
\end{example}

\begin{remark}[Relation between conditioning and trimming procedures]
    Let us comment on the connection between~\eqref{eq:autocor_unit_est_known_set} and the robust autocorrelation estimation procedures considered in the literature; see e.g.~\cite{Dur2015} for a comprehensive overview. In particular, at first glance, the sample conditional autocorrelation could be seen as a version of the sample autocorrelation when a trimming procedure is applied; see ``Trim'' estimator therein. Note that, however, the trimmed estimator is intended to estimate the (unconditional) autocorrelation while Theorem~\ref{th:consistency_ts} highlights the fact that the estimator given in~\eqref{eq:autocor_unit_est_known_set} converges to the conditional autocorrelation which could be different from its unconditional version; a similar observation applies to some other robust autocorrelation estimators. Also, the proposed methodology facilitates the use of any quantile split $0<p<q<1$ while in the trimming procedure, the typical choice consists of $q=1-p$ and $p$ close to $0$. This flexibility is particularly useful when a non-linear dependence structure is considered; see Section~\ref{S:GARCH} for details.
\end{remark}


\section{Applications of conditional correlations to serial dependence detection}\label{s4_intro}

In this section, we illustrate how to use the conditional autocorrelation to statistically test the presence of serial dependence in time series data.  For simplicity, we focus on lag $h=1$ and consider CACF and ACF functions at this argument. We consider three settings. First, in Section~\ref{S:MA1}, we study a moving average model with external noise. Next, in Section~\ref{S:GARCH} we consider a noise-corrupted generalized autoregressive conditional heteroskedasticity model. Finally, in Section~\ref{S:financial_auto}, we study the empirical financial data. It should be noted that models with external noise are widely discussed in the literature for various {\it undistorted} models and different types of external noise see, e.g., \cite{esfandiari,diversi1,parma_ao0,parma_an1,Zul_nowa}. These models are particularly useful in industrial applications where recorded signals are often corrupted by noise stemming from disruptions in the measuring systems or other external factors influencing the observations. External noise can also naturally occur in financial data due to influences from other markets. Usually, the presence of such noise makes identifying autodependence a challenging task. The proposed approach allows for more efficient identification of dependencies in time series in the presence of additional noise.  

Before we proceed, let us describe a testing framework used in this section. For each model, we use the conditional autocorrelation with lag $h=1$ to statistically check if an underlying stationary time series $\{X_t\}_{t \in \mathbb{Z}}$ consists of independent observations. To do this, given a fixed quantile split $0<p<q<1$ and the corresponding set $A$ defined in~\eqref{eq:set_A_autocor}, we simulate the distribution of $\hat\rho_A(1)$ under the null independence hypothesis. This distribution is based on $N\in\bN$ independent Monte Carlo samples of size $n\in\bN$; each sample consists of i.i.d. random variables with a cumulative (known) distribution function $F$. 
Then, given a type I error rate $ \alpha_{I} \in (0,1)$, we define a two-sided rejection region $R=(-\infty,\hat{Q}_{\hat\rho_A(1)}(\alpha_{I} /2)] \cup [\hat{Q}_{\hat\rho_A(1)}(1-\alpha_I  /2),\infty)$, where $\hat{Q}_{\hat\rho_A(1)}$ is the empirical quantile function of $\hat\rho_A(1)$. Finally, we reject the independence hypothesis at the level $\alpha_{I}$ if the value of $\hat\rho_A(1)$ on a tested sample belongs to $R$. In this paper, if not stated otherwise, we set $\alpha_{I}=0.05$. Also, note that this procedure in its simplest form requires that the stationary distribution $F$ is known; in Section~\ref{S:financial_auto} we show how to extend our framework to the unknown $F$ setting based on the bootstrap procedure.

\begin{remark}[Quantile set selection]

The proposed testing method requires a suitable choice of the quantile split $0<p<q<1$. This makes the framework flexible enough to capture various forms of dependence in the models. In particular, if the data is corrupted by some symmetric heavy-tail noise, it is reasonable to set $p$ close to $0$ and $q=1-p$ as this effectively works as a robust noise filter; we refer to Section~\ref{S:MA1} for a detailed discussion. Next, if the data follows from the centred distribution with a non-linear dependence structure, it is advisable to consider an asymmetric set $A$; see Section~\ref{S:GARCH} for details.
\end{remark}


\subsection{Serial dependence in MA(1) model with external noise}\label{S:MA1}

In this section, we consider a moving average model with an additional external noise. The {\it undistorted} lag 1 moving average (MA(1)) process $\{Z_t\}_{t\in \mathbb{Z}}$ is defined by the formula
\begin{align}\label{eq:MA_def}
    Z_t := \frac{\theta}  {\sqrt{1+\theta^2}}\varepsilon_{t-1}+\frac{1}{\sqrt{1+\theta^2}} \varepsilon_t, \quad t\in \mathbb{Z},
\end{align} 
where $\theta\in \mathbb{R}$ and $\varepsilon_t \sim \mathcal{N}(0,1)$ is the standard Gaussian white noise process. Next, we externally add a noise (independent from $\{\varepsilon_t\}_{t\in \mathbb{Z}}$) to this process. In fact, we consider two types of external noise processes $\{\psi^i_t\}_{t\in\mathbb{Z}}$, $i=1,2$, referred as discrete-type and continuous-type external noise. More specifically, for $i=1,2$ and $t\in \mathbb{Z}$, the random variables $\psi^i_t$ are i.i.d. and satisfy
\begin{align}\label{eq:discretenoise}
   \psi^1_t =
  \begin{cases}
        r   & \quad \text{with probability } \frac{P}{2}, \\
    -r  & \quad \text{with probability } \frac{P}{2} ,\\
    0 & \quad \text{with probability } 1-P,
  \end{cases} \quad\quad\textrm{and}\quad\quad \psi^2_t\sim S(\alpha,c),
\end{align}
where $r>0$ is the symmetric jump size, $P\in (0,0.5)$ is the jump probability, and $S(\alpha,c)$ is a symmetric $\alpha$-stable distribution with shape $\alpha\in (0,2]$ and scale $c>0$; note that for the stable distribution we use {\it 0-parametrization} in the sense of Section 1.3 from \cite{Nolan2020}. Finally, for $i=1,2$, we define a distorted time series {$\{X^i_t\}_{t\in \mathbb{Z}}$} given by
\begin{eqnarray}\label{additive}
X^i_t:=Z_t+ \psi_t^i, \quad t\in \mathbb{Z}.
\end{eqnarray}
Using standard results one can see that this process is strictly stationary (with the one-dimensional stationary distribution given as a convolution of a standard normal distribution and the distribution of $\psi_1^i$), but the observations are not independent (for $\theta\neq 0$). In the following, we check if the conditional autocorrelations could detect this serial dependence and confront their performance with a standard autocorrelation measure. 

Before we proceed, let us discuss in more detail the rationale behind the construction of the processes $\{X^i_t\}_{t\in\mathbb{Z}}$, $i=1,2$. First, note that in~\eqref{eq:MA_def} we used a non-standard parameterization of MA(1) so that the variance of $Z_t$, for any $t\in \mathbb{Z}$, is equal to one. This is a purely technical transformation which facilitates an easier comparison with the relative magnitude of the external noise $\{\psi_t^i\}_{t\in \mathbb{Z}}$. Also, it should be noted that in this case, the null (independence) hypothesis corresponds to $\theta=0$. Indeed,  the process $\{Z_t\}_{t\in \mathbb{Z}}$ is i.i.d. if and only if $\theta=0$, and this characterisation transfers to  $\{X^i_t\}_{t\in\mathbb{Z}}$, $i=1,2$. Also, for any $\theta\in \mathbb{R}$, the stationary distributions of $\{Z_t\}_{t\in \mathbb{Z}}$ and $\{X^i_t\}_{t\in\mathbb{Z}}$, $i=1,2$, does not depend on $\theta$. Consequently, the one-dimensional stationary distributions of $\{X_t^i\}_{t \in \mathbb{Z}}$ under the null and the alternative hypothesis are the same. This facilitates dependence structure testing instead of distribution alignment testing.

To better visualise the impact of an external noise on MA(1) process, in Figure~\ref{MA1Out} we present an exemplary trajectory of the {\it pure} process $\{Z_t\}_{t\in \mathbb{Z}}$ (left panel) with the corresponding trajectories of the {\it distorted} processes $\{X_t^i\}_{t\in \mathbb{Z}}$, $i=1,2$ (middle and right panel, respectively). One can see that the process $\{X_t^1\}_{t\in \mathbb{Z}}$ preserves some characteristics of $\{Z_t\}_{t\in \mathbb{Z}}$ (although, there are some outliers) while the nature of $\{X_t^2\}_{t\in \mathbb{Z}}$ is visibly different from $\{Z_t\}_{t\in \mathbb{Z}}$.

\begin{figure}[htp!]
\begin{center}

\includegraphics[width=0.32\textwidth]{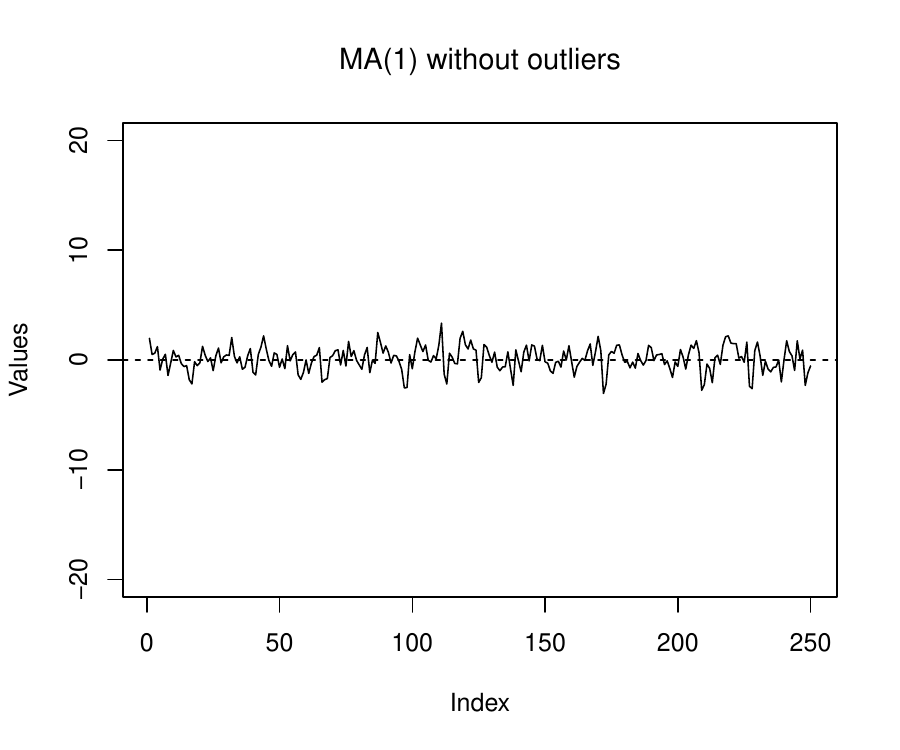}
\includegraphics[width=0.32\textwidth]{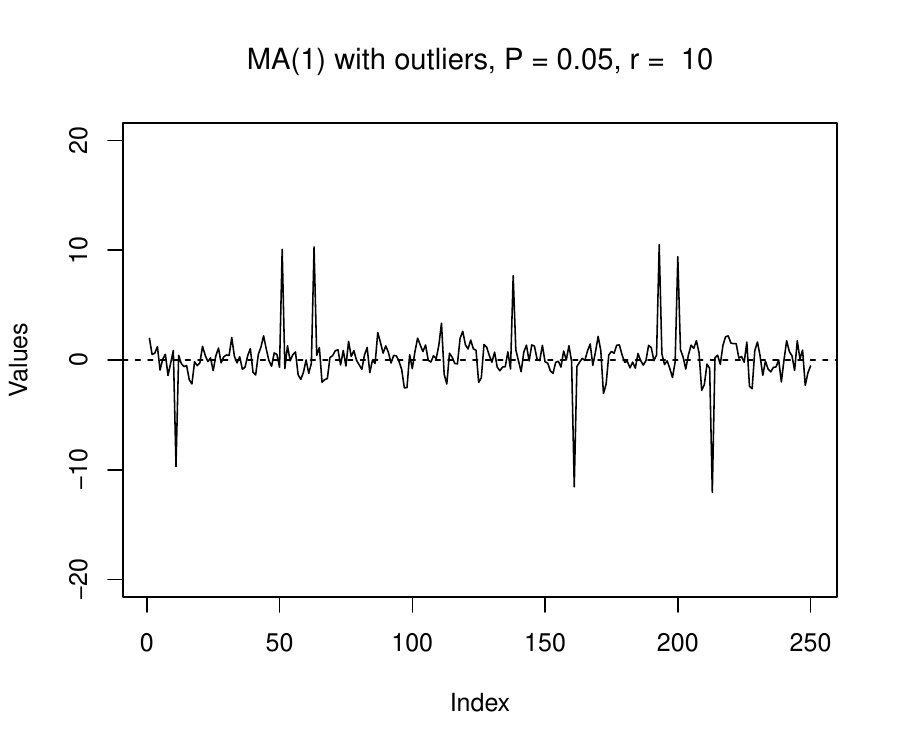}
\includegraphics[width=0.32\textwidth]{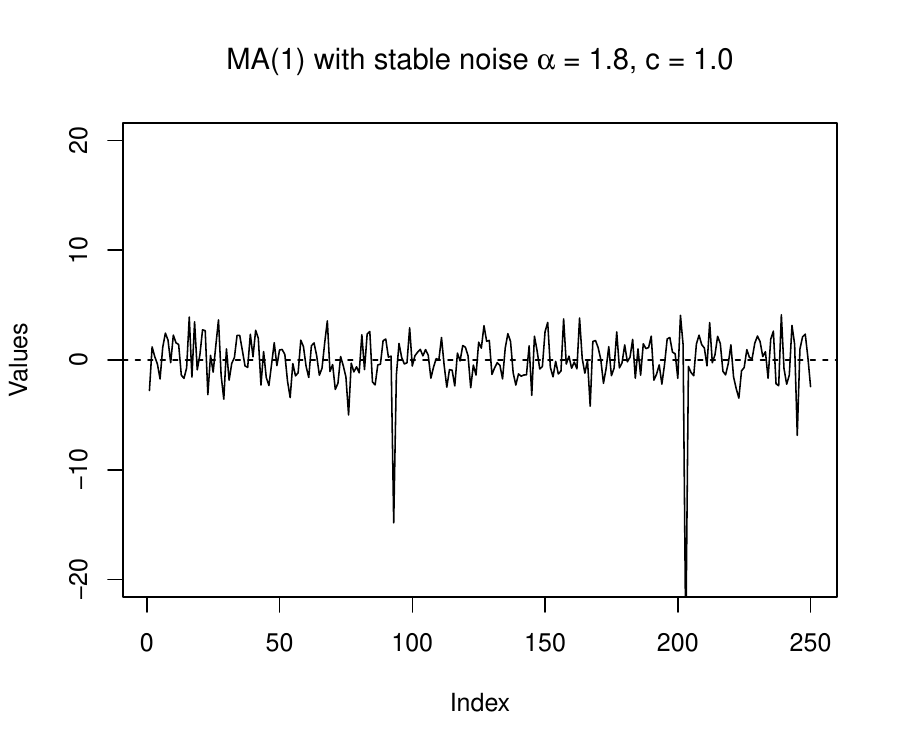}
\end{center}
\caption{The left panel shows an exemplary trajectory of MA(1) process with $\theta = 0.5$ while the middle and right panels show the corresponding trajectories of the processes $\{X_t^i\}_{t\in\mathbb{Z}}$ , for $i=1,2$, respectively. For   $\{\psi_t^1\}_{t\in\mathbb{Z}} $ we used the parameters $P = 0.05$ and $r = 10$, while for $\{\psi_t^2\}_{t\in\mathbb{Z}}$ we set $\alpha = 1.8$,  $c = 1.0$.}
\label{MA1Out}
\end{figure}

Let us now focus on assessing the proposed methodology for testing the serial dependence of the data. We will follow the testing framework described in Section~\ref{s4_intro} and use Monte Carlo simulations to estimate the power of the serial independence test based on conditional autocorrelation. In this section, we consider only symmetric quantile splits and set $q=1-p$ for some $p\in \{0.01, 0.1, 0.25\}$. The power of the tests based on $\rho_A(1)$ will be confronted with the standard autocorrelation $\rho(1)$. As already stated, the null hypothesis corresponds to $\theta=0$, i.e. independence of $\{X_t^i\}_{t\in \mathbb{Z}}$, while for the alternative hypothesis, we check various combinations of $\theta\in \{0.1,0.5, 0.9\}$, $P\in \{0.01, 0.15, 0.8\}$, $r\in \{1,8,15\}$, $\alpha\in \{1.05, 1.5, 2.0\}$, and $c\in \{0.1, 0.7, 1.5\}$. The distributions of the underlying test statistics for the null hypothesis are simulated based on $N=1000$ Monte Carlo samples with individual sample size $n=1000$. The power is estimated with the help of $1000$ samples from the alternative hypothesis with the type I error rate equal to 0.05.

The results are presented in Table~\ref{table:MA_jump}. 
\begin{table}[htp!]
\scalebox{0.75}{
\centering
\begin{tabular}{|ccc|ccc|c|}
  \hline
P & r & $\theta$ & $\rho_{(0.01,0.99)}(1)$ & $\rho_{(0.05,0.95)}(1)$ & $\rho_{(0.25,0.75)}(1)$ & $\rho(1)$ \\
  \hline
0.01 & 1.00 & 0.10 & 0.75 &  0.22 & 0.05 & \textbf{0.90} \\ 
0.01 & 1.00 & 0.50 & \textbf{1.00} & \textbf{1.00} & 0.17 & \textbf{1.00} \\ 
0.01 & 1.00 & 0.90 & \textbf{1.00} &  \textbf{1.00} & 0.26 & \textbf{1.00} \\ 
0.01 & 8.00 & 0.10 & \textbf{0.82} & 0.20 & 0.04&  0.52 \\ 
0.01 & 8.00 & 0.50 & \textbf{1.00} & \textbf{1.00} & 0.17&  \textbf{1.00} \\ 
0.01 & 8.00 & 0.90 & \textbf{1.00} & \textbf{1.00} & 0.27 & \textbf{1.00} \\ 
0.01 & 15.00 & 0.10 & \textbf{0.79} & 0.20 & 0.05 &0.12 \\ 
0.01 & 15.00 & 0.50 & \textbf{1.00} & \textbf{1.00} & 0.13& 0.95 \\ 
0.01 & 15.00 & 0.90 & \textbf{1.00} & \textbf{1.00} & 0.26& 0.99 \\ 
0.08 & 1.00 & 0.10 & 0.70 & 0.14 & 0.06 & \textbf{0.82} \\ 
0.08 & 1.00 & 0.50 & \textbf{1.00} & \textbf{1.00} & 0.18 & \textbf{1.00} \\ 
0.08 & 1.00 & 0.90 & \textbf{1.00} & \textbf{1.00}  & 0.26 & \textbf{1.00}\\ 
0.08 & 8.00 & 0.10 & 0.09 & \textbf{0.27} & 0.07& 0.09 \\ 
0.08 & 8.00 & 0.50 & 0.44 & 1.00 & 0.28 & \textbf{0.59} \\ 
0.08 & 8.00 & 0.90 & 0.59 & \textbf{1.00} & 0.46& 0.75 \\ 
0.08 & 15.00 & 0.10 & 0.06 & \textbf{0.26} & 0.07 & 0.06 \\ 
0.08 & 15.00 & 0.50 & 0.06 & \textbf{1.00} & 0.26& 0.11 \\ 
0.08 & 15.00 & 0.90 & 0.08 & 1.00 & 0.26&  0.13 \\ 
0.15 & 1.00 & 0.10 & 0.71 & 0.16 &0.06 & \textbf{0.83} \\ 
0.15 & 1.00 & 0.50 & \textbf{1.00} & 0.99 & 0.18 & \textbf{1.00} \\ 
0.15 & 1.00 & 0.90 & \textbf{1.00} & \textbf{1.00} & 0.28 & \textbf{1.00} \\ 
0.15 & 8.00 & 0.10 & 0.04 & \textbf{0.45} & 0.07 & 0.09 \\ 
0.15 & 8.00 & 0.50 & 0.09 & \textbf{1.00} & 0.34 & 0.20 \\ 
0.15 & 8.00 & 0.90 & 0.10 & \textbf{1.00} & 0.56 & 0.28 \\ 
0.15 & 15.00 & 0.10 & 0.05 & \textbf{0.44} & 0.06& 0.05 \\ 
0.15 & 15.00 & 0.50 & 0.05 & \textbf{1.00}& 0.32 & 0.05 \\ 
0.15 & 15.00 & 0.90 & 0.06 & \textbf{1.00} & 0.55&  0.09 \\ 
   \hline
\end{tabular}
\hfill
\centering
\,\,\,\,\,\,\,\,
\begin{tabular}{|ccc|ccc|c|}
  \hline
$\alpha$ & c & $\theta$ & $\rho_{(0.01,0.99)}(1)$ & $\rho_{(0.05,0.95)}(1)$ & $\rho_{(0.25,0.75)}(1)$ & $\rho(1)$ \\
  \hline
1.05 & 0.10 & 0.10 & \textbf{1.00}  & 0.92 & 0.08 &  0.63 \\
1.05 & 0.10 & 0.50 & \textbf{1.00} &  \textbf{1.00} & 0.10 & 0.76 \\
1.05 & 0.10 & 0.90 & \textbf{1.00} & \textbf{1.00} & 0.22 & 0.80 \\
1.05 & 0.70 & 0.10 & \textbf{0.99} & 0.92 & 0.06 &  0.21 \\
1.05 & 0.70 & 0.50 & \textbf{1.00} & \textbf{1.00} & 0.14 & 0.33 \\
1.05 & 0.70 & 0.90 & \textbf{1.00} & \textbf{1.00} & 0. 21 & 0.42 \\
1.05 & 1.50 & 0.10 & 0.71 & \textbf{0.86} & 0.06 & 0.07 \\
1.05 & 1.50 & 0.50 & 0.98 & \textbf{1.00} & 0.14 & 0.10 \\
1.05 & 1.50 & 0.90 & \textbf{1.00} & \textbf{1.00} & 0.21 & 0.14 \\
1.50 & 0.10 & 0.10 & \textbf{1.00} & 0.92 & 0.09 & 0.99 \\
1.50 & 0.10 & 0.50 & \textbf{1.00} & \textbf{1.00} & 0.10& 0.99 \\
1.50 & 0.10 & 0.90 & \textbf{1.00} & \textbf{1.00}& 0.22 & \textbf{1.00} \\
1.50 & 0.70 & 0.10 & \textbf{1.00} &  0.89 & 0.08 &  0.92 \\
1.50 & 0.70 & 0.50 & \textbf{1.00} & \textbf{1.00} & 0.11&  0.96 \\
1.50 & 0.70 & 0.90 & \textbf{1.00} & \textbf{1.00} & 0.19 & 0.98 \\
1.50 & 1.50 & 0.10 & \textbf{0.98} & 0.80 & 0.05 & 0.72 \\
1.50 & 1.50 & 0.50 & \textbf{1.00} & \textbf{1.00} & 0.12 & 0.88 \\
1.50 & 1.50 & 0.90 & \textbf{1.00} & \textbf{1.00} & 0.14 & 0.90 \\
2.00 & 0.10 & 0.10 & \textbf{1.00} & 0.91 & 0.07 & \textbf{1.00} \\
2.00 & 0.10 & 0.50 & \textbf{1.00} & \textbf{1.00} & 0.11 & \textbf{1.00} \\
2.00 & 0.10 & 0.90 & \textbf{1.00} &  \textbf{1.00} & 0.22 & \textbf{1.00} \\
2.00 & 0.70 & 0.10 & \textbf{1.00} & 0.93 & 0.06&  \textbf{1.00} \\
2.00 & 0.70 & 0.50 & \textbf{1.00} & \textbf{1.00} & 0.12 & \textbf{1.00} \\
2.00 & 0.70 & 0.90 & \textbf{1.00} & \textbf{1.00} & 0.18 & \textbf{1.00} \\
2.00 & 1.50 & 0.10 & 0.99& 0.82 & 0.08& \textbf{1.00} \\
2.00 & 1.50 & 0.50 & \textbf{1.00} &  \textbf{1.00} & 0.10 & \textbf{1.00} \\
2.00 & 1.50 & 0.90 & \textbf{1.00} &\textbf{1.00} & 0.16&  \textbf{1.00} \\
   \hline
\end{tabular}
\caption{The empirical power of the underlying tests for the MA(1) model with discrete noise (left panel) and the $\alpha-$stable noise (right panel) for multiple choices of parameters. }\label{table:MA_jump}
}
\end{table}
As expected, one can see that the power of all tests is increasing in $\theta$ as larger $\theta$ implies a stronger correlation in~\eqref{eq:MA_def}. For the discrete external noise $\{\psi_t^1\}_{t\in \mathbb{Z}}$, the power is also decreasing in the probability of outlier occurrence $P$ and in outlier size $r$. For the $\alpha-$stable external noise $\{\psi_t^2\}_{t\in \mathbb{Z}}$, the power is decreasing in the scale parameter $c$ and increasing in the stability parameter $\alpha$; note that the smaller $\alpha$, the fatter the tail of the distribution. Next, one can see that the serial dependence tests based on the conditional autocorrelation outperform  $\rho(1)$ (or obtain the same power) in almost all cases. The results are particularly striking for moderate values of $\alpha=1.5$ and $c=0.7$, where for any considered $\theta$, the empirical power of $\rho_{(0.01,0.99)}(1)$ equals 100\%, while the power of $\rho(1)$ is 54\%. However, one can notice that if the outliers are relatively rare (small $P$ or big $\alpha$) or have small magnitude (small $r$ or $c$), the power of $\rho(1)$ could be higher than the power of $\rho_A(1)$. This is expected as the outlier filtering mechanism embedded in the construction of $\rho_A(1)$ may lead to the loss of some information contained in the discarded observations. Still, one can account for this loss by choosing a smaller quantile split parameter $p$ (and bigger $q$).
Clearly, if the outliers are more common or with bigger magnitude, it is advisable to use bigger $p$ and smaller $q$ as this facilitates more efficient filtering; see e.g. $P=0.08$ and $r=15$, where the power of $\rho_{(0.1,0.9)}(1)$ is 4-10 times higher the power of $\rho_{(0.01,0.99)}(1)$ (and $\rho(1)$).

To better visualise the effects of the alternative distribution parameter choice on the test power, we focus on the conditional autocorrelation with $p=0.01$ and $q=0.99$. In Figure~\ref{plot:MA1Out} we present the results for the model with discrete noise $\{\psi_t^1\}_{t\in \mathbb{Z}}$ while in Figure~\ref{plot:carpet_stable} we show the power for the model with stable noise $\{\psi_t^2\}_{t\in \mathbb{Z}}$. In each case, we use $\theta=0.5$ in the construction of MA(1) model and compare the results with the serial dependence test based on the standard autocorrelation.

\begin{figure}[htp!]
\begin{center}
\includegraphics[width=0.28\textwidth]{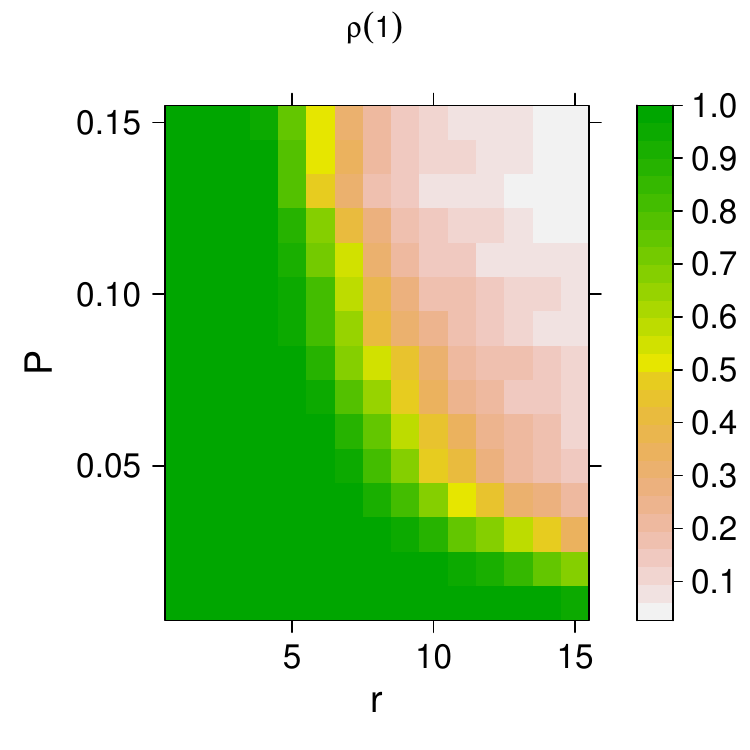}
\includegraphics[width=0.28\textwidth]{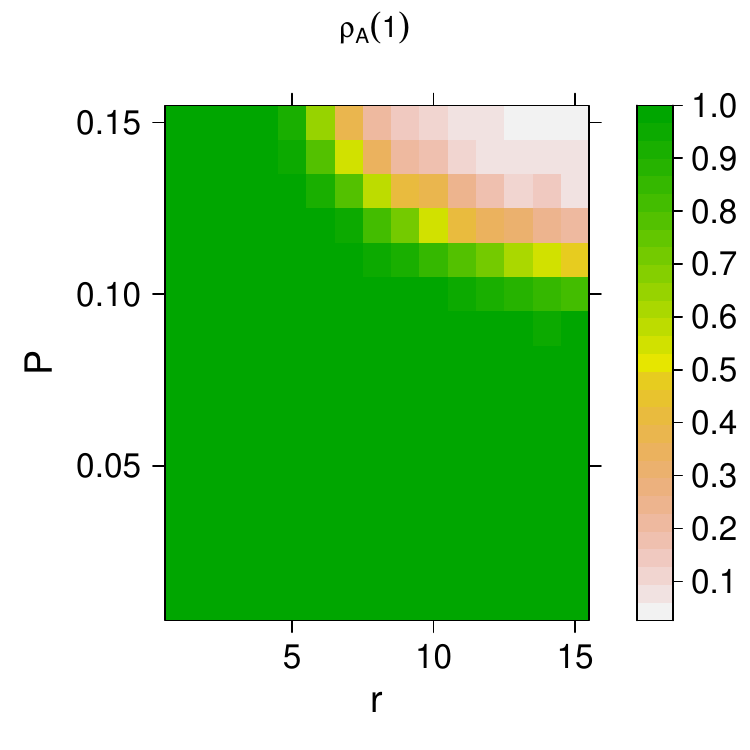}
\end{center}
\caption{Empirical power of the autocorrelation (left panel) and conditional autocorrelation (right panel) for the MA(1) model with $\theta=0.5$ and the discrete noise $\{\psi_t^1\}_{t\in\mathbb{Z}}$ with $r\in (0,15)$ and $P\in (0,0.15)$. The results are based on 1000 Monte Carlo samples with size $n=1000$. In the conditional autocorrelation we set $p=0.01$ and $q=0.99$.
}\label{plot:MA1Out}
\end{figure}
\begin{figure}[htp!]
\begin{center}
\includegraphics[width=0.28\textwidth]{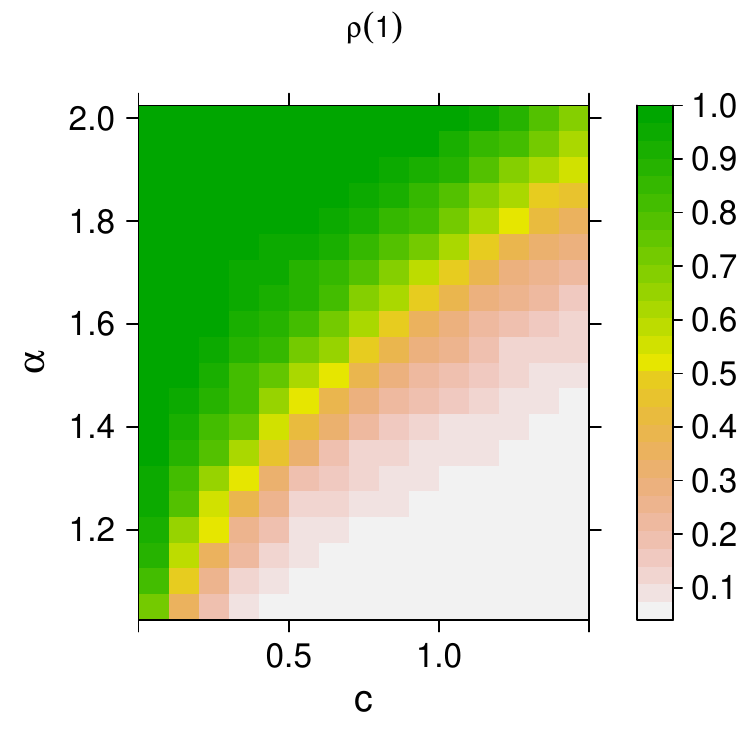}
\includegraphics[width=0.28\textwidth]{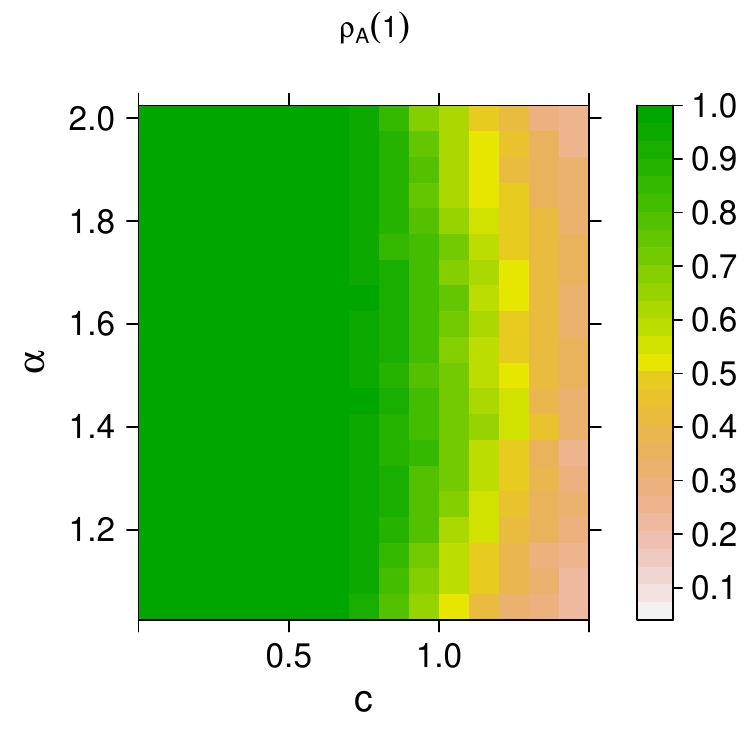}
\end{center}
\caption{
Empirical power of the autocorrelation (left panel) and conditional autocorrelation (right panel) for the MA(1) model with $\theta=0.5$ and the $\alpha$-stable noise $\{\psi_t^2\}_{t\in\mathbb{Z}}$ with $c\in (0.1,1.5)$ and $\alpha\in (1,2)$. The results are based on 1000 Monte Carlo samples with size $n=1000$. In the conditional autocorrelation we set $p=0.01$ and $q=0.99$.
}\label{plot:carpet_stable}
\end{figure}

From Figure~\ref{plot:MA1Out} one can see that that the test based on the conditional correlation outperforms the (unconditional) autocorrelation in almost all cases. In fact, even for relatively high jump sizes $r>10$ (and $P\in (0.05,0.1)$), the test based on $\rho(1)$ detects the serial dependence in 20-30\% of the samples while the test based on $\rho_A(1)$ is able to correctly identity the lack of independence in almost all cases. As already observed, the power decreases with the jump size $r$ and the jump probability $P$. Still, one can account for this effect by increasing the quantile split parameter $p$; see Table~\ref{table:MA_jump} for details. Indeed, with larger $p$, the conditional autocorrelation is able to more effectively filter out the noise which results in higher power.

Figure~\ref{plot:carpet_stable} confirms a good performance of the serial dependence test based on the conditional autocorrelation. In fact, it can be seen that while the power of the unconditional autocorrelation test drops significantly as the stability parameter $\alpha$ goes from 2 (Gaussian distribution) to 1 (Cauchy distribution), the results for the conditional correlation are almost unaffected by $\alpha$. For both tests, the bigger the value of the scale parameter $c$, the smaller the power. This could be linked to the fact that large $c$ corresponds to large values of the outliers which more severely corrupt the original MA(1) process. Finally, it should be noted that in the close to Gaussian case ($\alpha$ close to 2) and relatively big values of $c$ (e.g. $c>0.9$), the unconditional autocorrelation outperforms $\rho_A(1)$. This could be explained by the construction of the conditional autocorrelation statistic. Indeed, this test statistic focuses on some subset of data which facilitates the filtering of outliers (when they are present) but reduces the effective sample size. Still, for such a regular case ($\alpha$ close to 2), one can improve the performance by choosing $p$ close to 0 and $q$ close to 1; see Table~\ref{table:MA_jump} for details.

\FloatBarrier

\subsection{Serial dependence in GARCH(1,1)  model with external noise}\label{S:GARCH}

In this section, we consider a generalized autoregressive conditional heteroskedasticity (GARCH) model with external noise. The model construction follows the logic described in Section~\ref{S:MA1}. Here, the {\it undistorted} process $\{Z_t\}_{t\in \mathbb{Z}}$ is GARCH(1,1) process given by the formula
\begin{align}\label{eq:garch_model}
    \begin{cases}
\displaystyle         Z_t := \frac{\sigma_t }{\sqrt{\omega_0(1-\omega_1-\omega_2)^{-1}}}\varepsilon_t, \quad t\in \mathbb{Z},\\
\displaystyle          \sigma^2_t := \omega_0 + \omega_1 \varepsilon_{t-1}^2+\omega_2\sigma^2_{t-1}, \quad t\in \mathbb{Z},
    \end{cases}
\end{align}
where $\omega_i \geq 0$, for $i \in \{0,1,2\}$, $\omega_1 + \omega_2 < 1$, and $\varepsilon_t \sim \mathcal{N}(0,1)$ is the standard Gaussian white noise. As in the previous case, we consider two types of external noise processes given by~\eqref{eq:discretenoise}.The resulting processes $\{X^i_t\}_{t\in \mathbb{Z}}$, $i=1,2$, are defined as in~\eqref{additive}.  
Exemplary trajectories for the {\it pure} process $\{Z_t\}_{t\in \mathbb{Z}}$ and the corresponding processes $\{X^i_t\}_{t\in \mathbb{Z}}$ are presented in Figure~\ref{GARCHtajectory}.

\begin{figure}[htp!]
\begin{center}
\includegraphics[width=0.32\textwidth]{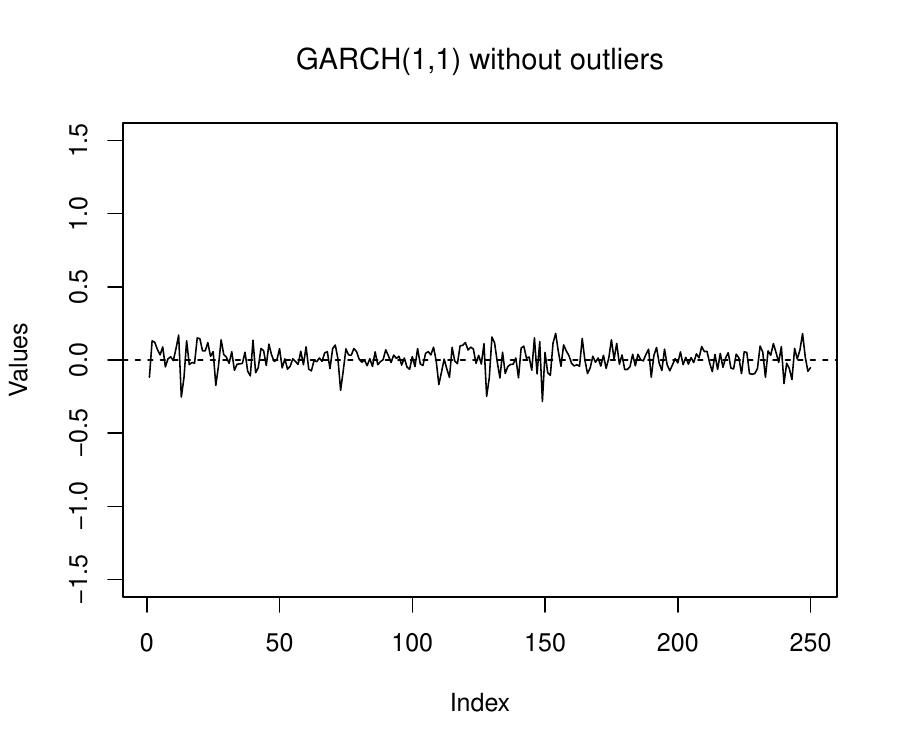}
\includegraphics[width=0.32\textwidth]{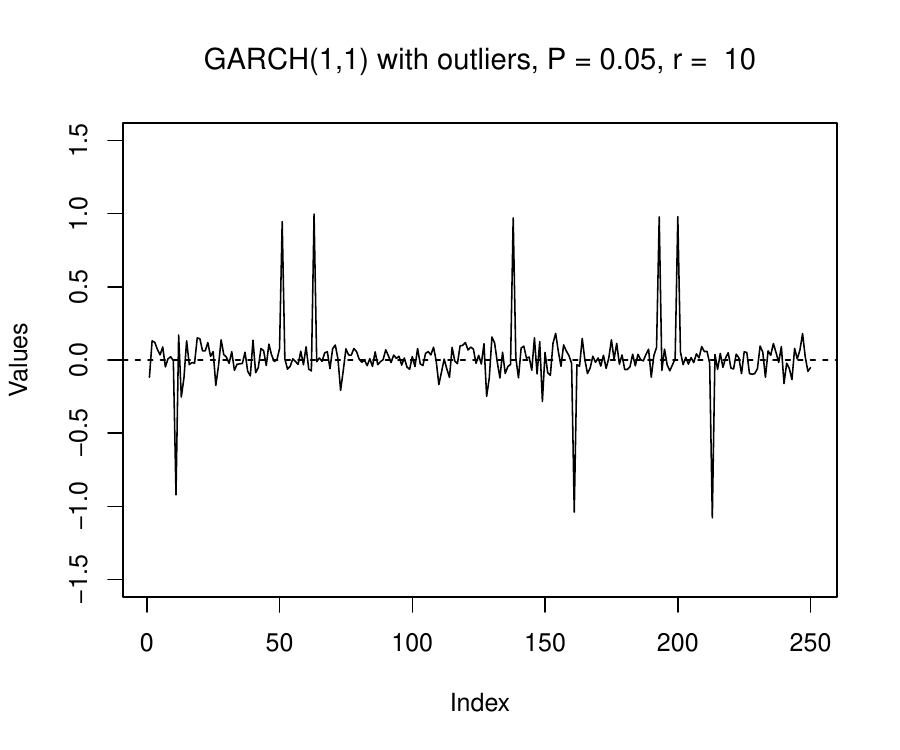}
\includegraphics[width=0.32\textwidth]{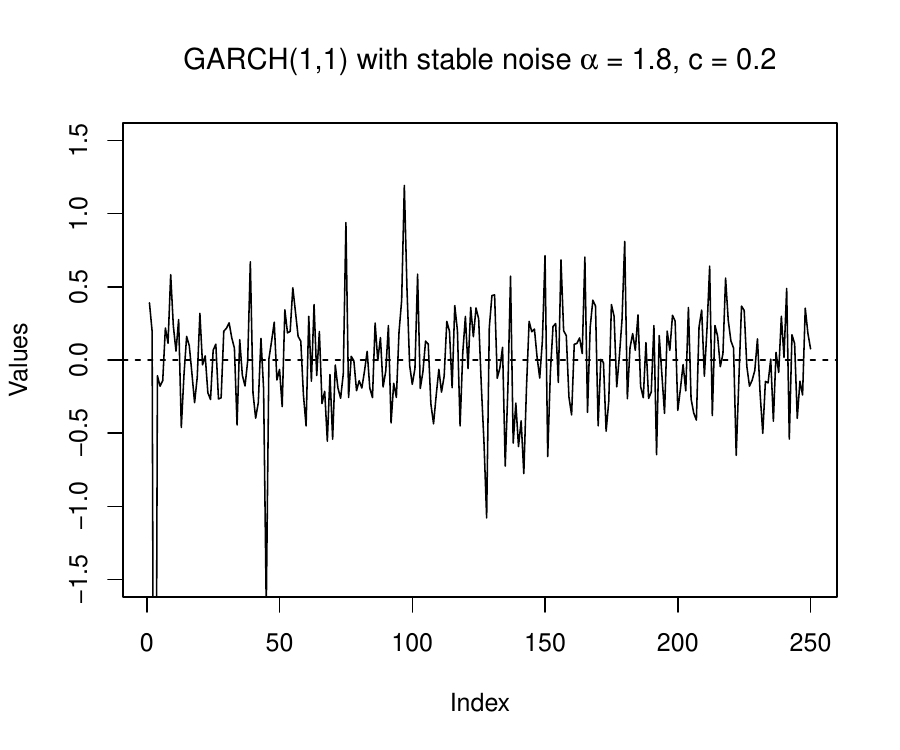}
\end{center}
\caption{The left panel shows an exemplary trajectory of GARCH(1,1) process with $\omega_0 = 0.001,
\omega_1 = 0.2$ and $\omega_2 = 0.6$ while the middle and right panels show the corresponding trajectories of the processes $\{X_t^i\}_{t\in\mathbb{Z}}$, for $i=1,2$, respectively. For   $\{\psi_t^1\}_{t\in\mathbb{Z}} $ we used the parameters $P = 0.05$ and $r = 10$, while for $\{\psi_t^2\}_{t\in\mathbb{Z}} $ we set $\alpha = 1.8$,  $c = 1.0$.}\label{GARCHtajectory}
\end{figure}

As already stated, the construction of the processes considered here is similar to the ones from Section~\ref{S:MA1}. In particular, the factor $\left(\omega_0(1-\omega_1-\omega_2)^{-1}\right)^{-1/2}$ in the definition of $Z_t$ results in the unit variance of the one-dimensional stationary distribution of $Z_t$; see a similar discussion for the parameterization of the MA(1) process defined in~\eqref{eq:MA_def}. Note that, however, the undistorted MA(1) process from Section~\ref{S:MA1} exhibits a linear dependence structure while the dependence in GARCH(1,1) process from~\eqref{eq:garch_model} is highly non-linear. This dependence structure makes serial dependence detection more challenging.

We use Monte Carlo simulations to assess the performance of the conditional autocorrelations. Since for GARCH process, the conditional autocorrelations with any $p\in (0,0.5)$ and $q=1-p$ are equal to zero, in this section, we consider nonsymmetric quantile splits, i.e. $\rho_{(p,q)}$ with $ p\in \{0.01, 0.05\}$ and $q \in \{0.5, 0.65, 0.9\} $. The power of the test based on $ \rho_A(1) $ is confronted with the standard autocorrelation $\rho(1) $ and with the autocorrelation calculated on the squared sample (denoted by $\rho(1)(x^2)$); note that the use of the last measure is a standard procedure for GARCH processes, see e.g. Chapter 7 in~\cite{brockwell2016introduction} for a more detailed discussion. As previously, the null hypothesis states that the underlying time series consists of i.i.d. observations. For the alternative hypothesis, we check various values of  $\omega _1 \in \{0.2,0.4, 0.6\}$. To reduce the number of parameters, we also set $\omega_2 = 0.9-\omega_1$ and $\omega_0 = 0.001$. For the external noise parameters, we take $P\in \{0.01, 0.08, 0.15\}$, $r\in \{1,8,15\}$, $\alpha\in \{1.05, 1.5, 2.0\}$, and $c\in \{0.1, 0.3, 0.\}$. The distributions of the underlying test statistics for the null (independence) hypothesis are simulated based on $N=1000$ Monte Carlo samples with individual sample size $n=1000$. Because of the lack of a closed-form formula for the univariate (unconditional) distribution of GARCH process, to simulate i.i.d. sample from the distribution of $Z_t$, we used the fact that random variables defined recursively by~\eqref{eq:garch_model} converge to the stationary distribution of GARCH process; see e.g.~\cite{Lin2009} for details. More specifically, we picked some initial values $\varepsilon_0$ and $\sigma^2_0$, simulated $\{\varepsilon_t\}_{t \in \mathbb{Z}}$ as i.i.d. $\mathcal{N}(0,1)$ random variables, and obtain one realisation of $Z_t$ by using~\eqref{eq:garch_model} recursively for $k=10000$ times; note that $k$ could be seen as a burn-in period parameter for this particular Markov Chain Monte Carlo simulation. This procedures was repeated $n\times N$ times to get $N$ samples with $n$ independent observations.
The power is estimated with the help of $1000$ samples from the alternative hypothesis with the type I error rate equal to $0.05$; each (dependent) sample is simulated directly from~\eqref{eq:garch_model}. 

The results are presented in Table~\ref{table:GARCH}.
\begin{table}[htp!]
\centering 
\scalebox{0.7}{
\begin{tabular}{|ccc|cccccc|cc|}
  \hline
 P & r & $w_1$ & $\rho_{(0.01,0.5)}(1)$ & $\rho_{(0.01,0.65)}(1)$ & $\rho_{(0.01,0.9)}(1)$ & $\rho_{(0.05,0.5)}(1)$ & $\rho_{(0.05,0.65)}(1)$ & $\rho_{(0.05,0.9)}(1)$ & $\rho(1)$ & $\rho(1)(x^2)$ \\ 
  \hline
0.01    & 1.00  & 0.20 & 0.68   &  0.69   & 0.24 & 0.22 & 0.20 & 0.05 & 0.13   & \textbf{0.94}   \\
0.01    & 1.00  & 0.40 & 0.97 & \textbf{1.00} & 0.62 & 0.58 & 0.64 & 0.09   & 0.26   & \textbf{1.00}   \\
0.01    & 1.00  & 0.60 & \textbf{1.00}  & \textbf{1.00} & 0.85 & 0.84 & 0.87 & 0.12 & 0.37   & \textbf{1.00}   \\
0.01    & 8.00  & 0.20 & \textbf{0.74} & 0.73 & 0.29 & 0.20& 0.19 & 0.05 & 0.07   & 0.04   \\
0.01    & 8.00  & 0.40 & 0.98 & \textbf{0.99} & 0.69 & 0.56 & 0.65 & 0.07  & 0.08   & 0.05   \\
0.01    & 8.00  & 0.60 & \textbf{1.00}  & \textbf{1.00} & 0.85 & 0.81 & 0.85 & 0.10 & 0.10   & 0.06   \\
0.01    & 15.00 & 0.20 & \textbf{0.73} & \textbf{0.73} & 0.30 & 0.19 & 0.19 & 0.04  & 0.08   & 0.04   \\
0.01    & 15.00 & 0.40 & 0.98  & \textbf{0.99} & 0.72 & 0.58 & 0.61 & 0.08 & 0.08   & 0.03   \\
0.01    & 15.00 & 0.60 & \textbf{1.00}  & \textbf{1.00} & 0.85 & 0.81 & 0.87 & 0.10 & 0.08   & 0.03   \\
0.08    & 1.00  & 0.20 & 0.27 & \textbf{0.42} & 0.14 & 0.34 & 0.41 & 0.07   & 0.07   & 0.41   \\
0.08    & 1.00  & 0.40 & 0.51  & 0.74 & 0.28 & 0.85 & \textbf{0.90} & 0.18 & 0.15   & 0.76   \\
0.08    & 1.00  & 0.60 & 0.57 & 0.82 & 0.40 & \textbf{0.98} & \textbf{0.98} & 0.31  & 0.21   & 0.85   \\
0.08    & 8.00  & 0.20 & 0.04  & 0.04 & 0.05 & \textbf{0.49} & 0.48 & 0.12 & 0.05   & 0.04   \\
0.08    & 8.00  & 0.40 & 0.04 & 0.04 & 0.06 & \textbf{0.91} & \textbf{0.91} & 0.27  & 0.06   & 0.05   \\
0.08    & 8.00  & 0.60 & 0.04  & 0.04 & 0.06 & \textbf{0.97} & \textbf{0.97} & 0.42 & 0.06   & 0.06   \\
0.08    & 15.00 & 0.20 & 0.03  & 0.03 & 0.04 & 0.47 & \textbf{0.48} & 0.13 & 0.04   & 0.04   \\
0.08    & 15.00 & 0.40 & 0.02 & 0.02 & 0.03 & \textbf{0.90} & 0.89 & 0.29  & 0.04   & 0.05   \\
0.08    & 15.00 & 0.60 & 0.04 & 0.03 & 0.03 & \textbf{0.96} & \textbf{0.96} & 0.43   & 0.04   & 0.07   \\
0.15    & 1.00  & 0.20 & 0.12 & 0.19 & 0.06 & 0.38 & \textbf{0.47} & 0.08  & 0.05   & 0.23   \\
0.15    & 1.00  & 0.40 & 0.24 & 0.44 & 0.09 & 0.70 & \textbf{0.82} & 0.14   & 0.13   & 0.60   \\
0.15    & 1.00  & 0.60 & 0.24 & 0.49 & 0.17 & 0.78 & \textbf{0.93} & 0.23  & 0.17   & 0.71   \\
0.15    & 8.00  & 0.20 & 0.05 & 0.06 & 0.05 & 0.06 & \textbf{0.10} & 0.07   & 0.04   & 0.04   \\
0.15    & 8.00  & 0.40 & 0.04  & 0.05 & 0.05 & 0.08 & \textbf{0.14} & 0.07 & 0.04   & 0.04   \\
0.15    & 8.00  & 0.60 & 0.05 & 0.04 & 0.04 & 0.09 & \textbf{0.16} & 0.09  & 0.04   & 0.04   \\
0.15    & 15.00 & 0.20 & 0.05 & 0.05 & 0.05 & 0.06 & \textbf{0.07} & 0.06  & 0.04   & 0.05   \\
0.15    & 15.00 & 0.40 & 0.04  & 0.05 & 0.04 & 0.06 & \textbf{0.09} & 0.07 & 0.04   & 0.05   \\
0.15    & 15.00 & 0.60 & 0.03 & 0.04 & 0.04 & 0.07 & \textbf{0.09} & 0.07  & 0.05   & 0.06   \\  \hline
\end{tabular}}

\vspace{0.2cm}

\scalebox{0.7}{
\begin{tabular}{|ccc|cccccc|cc|}
  \hline
 $\alpha$ & c & $w_1$ & $\rho_{(0.01,0.5)}(1)$ & $\rho_{(0.01,0.65)}(1)$ & $\rho_{(0.01,0.9)}(1)$ & $\rho_{(0.05,0.5)}(1)$ & $\rho_{(0.05,0.7)}(1)$ & $\rho_{(0.05,0.9)}(1)$ & $\rho(1)$ & $\rho(1)(x^2)$ \\ 
  \hline
1.05  & 0.10 & 0.20 & 0.49 & \textbf{0.50} & 0.20 & 0.15 & 0.16 & 0.05  & 0.05   & 0.04  \\
1.05  & 0.10 & 0.40 & 0.88 & \textbf{0.89}& 0.42 & 0.40 & 0.41 & 0.06  & 0.09   & 0.05  \\
1.05  & 0.10 & 0.10 & 0.96 & \textbf{0.97} & 0.60 & 0.64 & 0.67 & 0.08 & 0.11  &  0.09  \\
1.05  & 0.30 & 0.20 & \textbf{0.15}  & 0.12 & 0.09 & 0.13 & \textbf{0.15} & 0.07  & 0.04   & 0.05  \\
1.05  & 0.30 & 0.40 & 0.30 & \textbf{0.38} & 0.22 & 0.30 & 0.29 & 0.08  & 0.04   & 0.04  \\
1.05  & 0.30 & 0.60 & 0.43 & \textbf{0.51} & 0.30 & 0.45 & 0.44 & 0.07  & 0.05   & 0.06  \\
1.05  & 0.50 & 0.20 & 0.06 & \textbf{0.08} & 0.06 & 0.04 & 0.06 & 0.07  & 0.06   & 0.05  \\
1.05  & 0.50 & 0.40 & 0.11  & 0.12 & 0.08 & 0.14 & \textbf{0.23} & 0.07 & 0.04   & 0.04  \\
1.05  & 0.50 & 0.60 & 0.12 & 0.17  & 0.11 & 0.15 & \textbf{0.31} & 0.10  & 0.04   & 0.04  \\
1.50  & 0.10 & 0.20 & 0.40  & 0.43 & 0.14 & 0.10 & 0.14 & 0.04 & 0.08   & \textbf{0.58}  \\
1.50  & 0.10 & 0.40 & 0.86  & \textbf{0.91} & 0.34 & 0.28& 0.13 & 0.06 & 0.19   & 0.71  \\
1.50  & 0.10 & 0.60 & 0.98  & \textbf{0.99} & 0.55 & 0.49& 0.62 & 0.06 & 0.31   & 0.77  \\
1.50  & 0.30 & 0.20 & 0.33  & \textbf{0.34} & 0.13 & 0.14 & 0.13 & 0.06 & 0.07   & 0.06  \\
1.50  & 0.30 & 0.40 & 0.67 & \textbf{0.72} & 0.26 & 0.21 & 0.22 & 0.08  & 0.11   & 0.21  \\
1.50  & 0.30 & 0.60 & 0.84 & \textbf{0.87} & 0.43 & 0.28 & 0.27 & 0.09  & 0.16   & 0.26  \\
1.50  & 0.50 & 0.20 & \textbf{0.18} & \textbf{0.18} & 0.10& 0.10 & 0.08 & 0.04  & 0.06   & 0.02  \\
1.50  & 0.50 & 0.40 & \textbf{0.39} & 0.38 & 0.19 & 0.16 & 0.13 & 0.08   & 0.09  & 0.06  \\
1.50  & 0.50 & 0.60 & 0.47 & \textbf{0.53} & 0.25 & 0.17 & 0.17 & 0.07   & 0.10   & 0.10  \\
2.00  & 0.10 & 0.20 & 0.47& 0.43 & 0.09 & 0.09 & 0.14 & 0.03   & 0.09   & \textbf{1.00}  \\
2.00  & 0.10 & 0.40 & 0.92 & 0.92 & 0.38 & 0.37 & 0.43 & 0.05   & 0.22   & \textbf{1.00}  \\
2.00  & 0.10 & 0.60 & 0.99 & \textbf{1.00}  & 0.61 & 0.56 & 0.65 & 0.08  & 0.34   & \textbf{1.00}  \\
2.00  & 0.30 & 0.20 & 0.29 & 0.30 & 0.09 & 0.14 & 0.10 & 0.05  & 0.07   & \textbf{0.98}  \\
2.00  & 0.30 & 0.40 & 0.66& 0.74 & 0.22 & 0.16 & 0.18 & 0.05   & 0.20   & \textbf{1.00}  \\
2.00  & 0.30 & 0.60 & 0.82 & 0.88 & 0.34 & 0.28 & 0.26 & 0.05   & 0.30   & \textbf{1.00}  \\
2.00  & 0.50 & 0.20 & 0.13 & 0.18 & 0.08 & 0.08 & 0.08 & 0.04   & 0.06   & \textbf{0.82}  \\
2.00  & 0.50 & 0.40 & 0.34 & 0.37 & 0.13 & 0.12 & 0.11 & 0.04   & 0.16   & \textbf{0.99}  \\
2.00  & 0.50 & 0.60 & 0.41 & 0.43 & 0.18 & 0.10 & 0.12 & 0.04   & 0.23   & \textbf{1.00}  \\ \hline
\end{tabular}
}
\caption{The empirical power of the underlying tests for the GARCH(1,1) model with discrete noise (top panel) and the stable noise (bottom panel) for multiple choices of parameters.}\label{table:GARCH}
\end{table}
Up to the simulation error, for each test, the power increases with $\omega_1$. Also, as in Section~\ref{S:MA1}, for the discrete external noise $\{\psi_t^1\}_{t\in \mathbb{Z}}$, the power decreases with respect to $P$ and $r$ while for the $\alpha-$stable external noise the power decreases with $c$ and increases with $\alpha$. Next, it can be noted that the serial dependence tests based on conditional autocorrelation outperform $\rho(1)$ in almost all cases. Also, these methods usually yield comparable or better results than $\rho(1)(x^2)$. In particular, for $\alpha = 1.5$, $c=0.1$, and $\omega_1 = 0.6$, the empirical power of 
$\rho_{(0.01,0.5)}(1)$ equals 80\%, while the power of $\rho(1)$ and $\rho(1)(x^2)$ are 16\% and 23\%, respectively. However, it is worth noting that if outliers are relatively rare (small $P$ or large $\alpha$) or have small magnitudes (small $r$ or $c$), the power of $\rho(1)(x^2)$ could surpass that of $\rho_A(1)$. This is also expected as the outlier filtering mechanism embedded in the construction of $\rho_A(1)$, similar to the MA example, may lead to the loss of information contained in the discarded observations. Nevertheless, this loss can be mitigated by selecting a smaller quantile split parameter $p$ and a larger $q$. Clearly, if outliers are more common or have greater magnitudes, it is advisable to use larger $p$ and smaller $q$, as this facilitates more efficient filtering; for instance, for $P=0.08$ and $r=15$, where the power of
$\rho_{(0.05,0.7)}(1)$
is much higher than the power of $\rho_{(0.01,0.7)}(1)$ (and $\rho(1)$, $\rho(1)(x^2)$).

Now, we focus on a more detailed assessment of the parameter impact on the test's performance. To do this, we consider the GARCH(1,1) model with the parameters $\omega_0 = 0.001$,  $\omega_1 = 0.6$ and $\omega_2 = 0.2$. Also, we fix the quantile split $p = 0.01$ and $q=0.65$, and focus on the effects of the additional noise parameters. The results are presented in Figure~\ref{plot:garch:par_additive} (discrete noise) and Figure~\ref{plot:garch_stable} ($\alpha-$stable noise).

\begin{figure}[htp!]
\begin{center}
\includegraphics[width=0.28\textwidth]{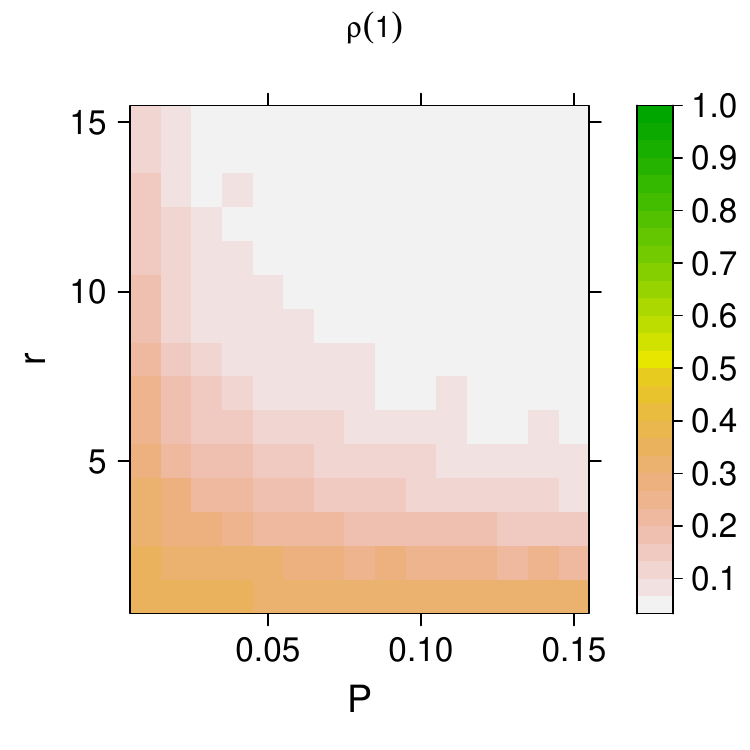}
\includegraphics[width=0.28\textwidth]{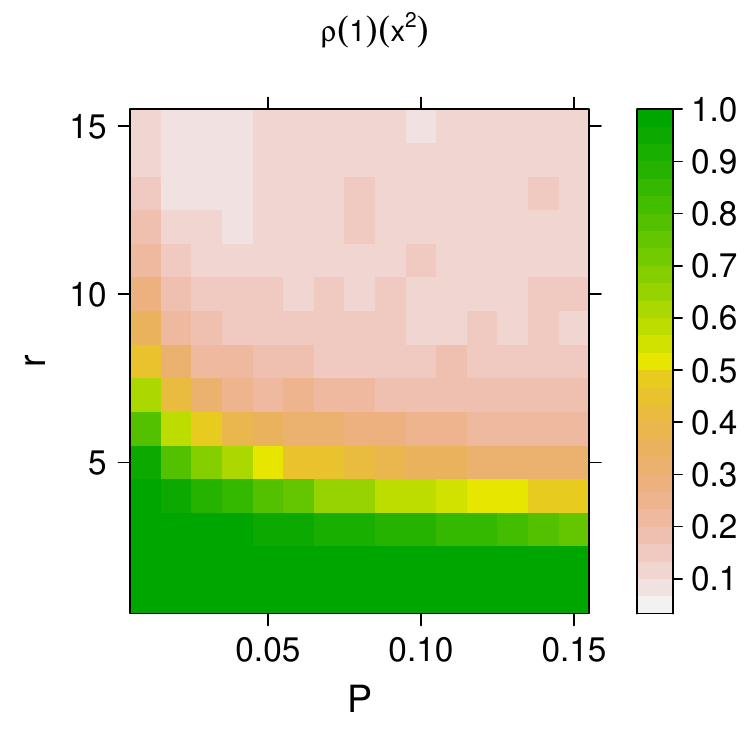}
\includegraphics[width=0.28\textwidth]{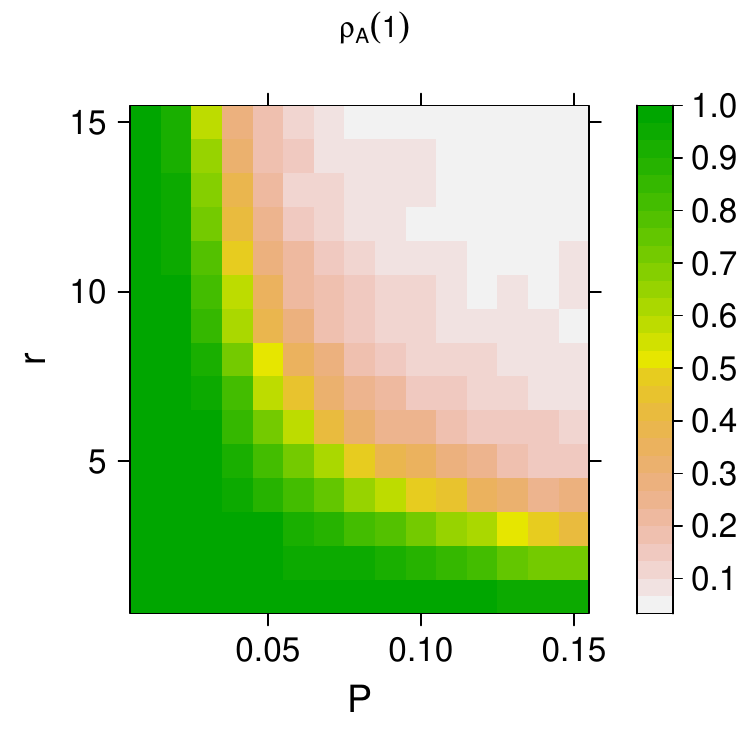}
\end{center}
\caption{
Empirical power of the autocorrelation (left panel), autocorrelation on squared data (middle panel), and conditional autocorrelation (right panel) for the GARCH(1,1) model with $\omega_0=0.001$, $\omega_1=0.6$ , $\omega_2=0.2$,  and the discrete noise $\{\psi_t^1\}_{t\in\mathbb{Z}}$ with $r\in (0,15)$ and $P\in (0,15)$. The results are based on $1000$ Monte Carlo samples with size $n=1000$. In the conditional autocorrelation, we set $p=0.01$ and $q=0.65$. 
}\label{plot:garch:par_additive}
\end{figure}

\begin{figure}[htp!]
\begin{center}
\includegraphics[width=0.28\textwidth]{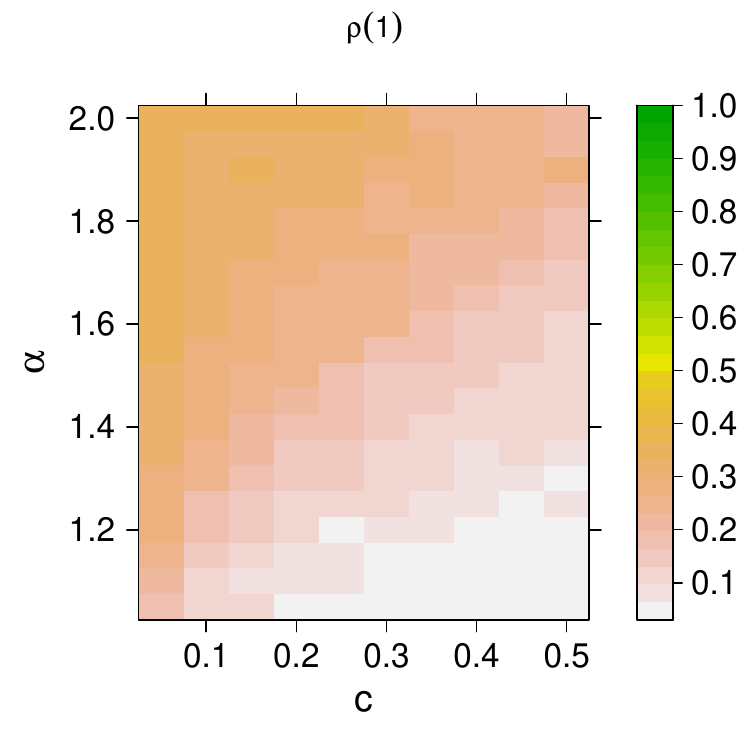}
\includegraphics[width=0.28\textwidth]{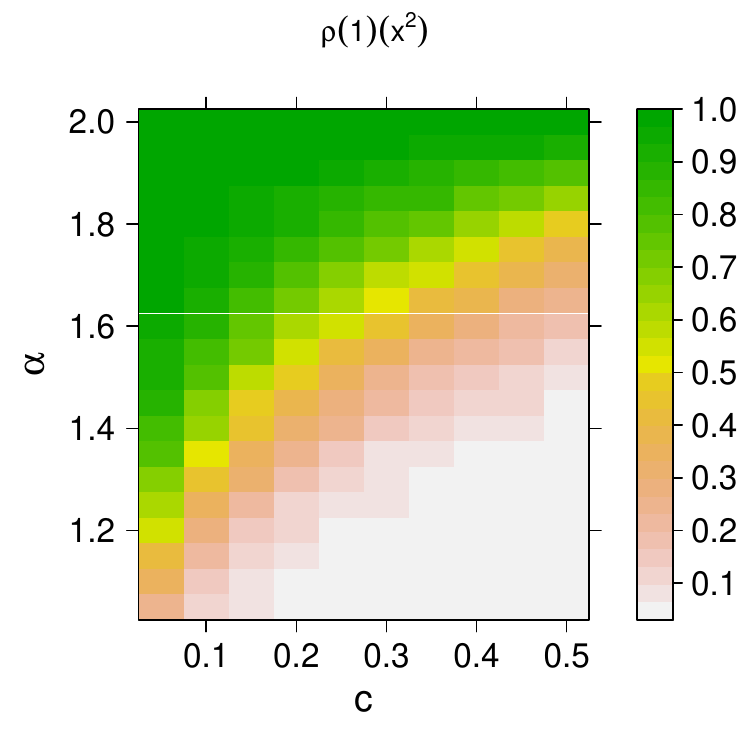}
\includegraphics[width=0.28\textwidth]{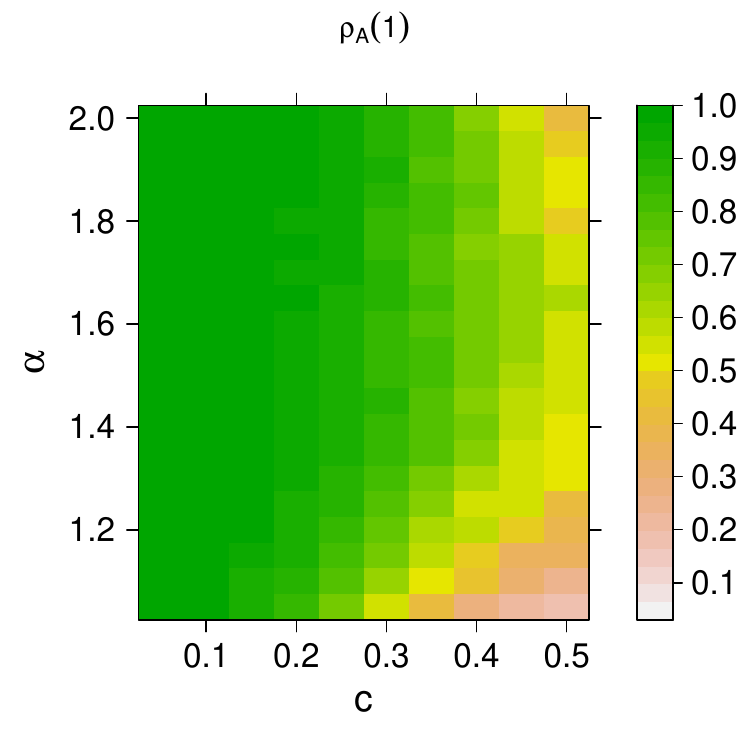}
\end{center}
\caption{
Empirical power of the autocorrelation (left panel), autocorrelation on squared data (middle panel), and conditional autocorrelation (right panel) for the GARCH(1,1) model with $\omega_0=0.001$, $\omega_1=0.6$ , $\omega_2=0.2$,  and the $\alpha-$stable noise $\{\psi_t^2\}_{t\in\mathbb{Z}}$ with $c\in (0,0.5)$ and $\alpha\in (1,2)$. The results are based on 1000 Monte Carlo samples with size $n=1000$. In the conditional autocorrelation we set $p=0.01$ and $q=0.65$. 
}\label{plot:garch_stable}
\end{figure}

Figure~\ref{plot:garch:par_additive} confirms that the autocorrelation for the squared data outperforms the classical autocorrelation for GARCH models. This is consistent with the observation that GARCH processes consist of uncorrelated but dependent observations. The power of $\rho(1)(x^2)$ is particularly high for small values of the additional noise noise $r$ (irregardless of the parameter $P$). Still, the power of the conditional autocorrelation is higher in the majority of the considered cases. This is particularly visible for small $P$ ($P<0.03$) and relatively big $r$ ($r>8$), where the power of $\rho_{(0.01,0.65)}(1)$ is close to 100\% while the power of $\rho(1)(x^2)$ is below 20\%.

Similar results hold for the model with additional $\alpha-$stable noise; see Figure~\ref{plot:garch_stable}. Again, the autocorrelation on the squared data outperforms the classical autocorrelation but even higher power is obtained by the conditional autocorrelation. Only if $\alpha$ is close to $2.0$ and $c$ is relatively high ($c>0.35$), the power of $\rho(1)(x^2)$ is higher than the power of $\rho_{(0.01,0.65)}(1)$. This could be associated with the fact that for such a choice of parameters, the additional noise is close to a very regular Gaussian white noise, where the usage of more standard tools is advisable.

\FloatBarrier

\subsection{Empirical analysis of S\&P500 stock market data}\label{S:financial_auto}

In this section, we apply the proposed methodology to the empirical data. More specifically, we consider adjusted logarithmic daily returns for S\&P 500 stocks in the period 01/01/2022--01/01/2024 and check if they exhibit some form of serial dependence. Due to the changes in the index composition, we obtained the full historical record for 498 companies with logarithmic daily returns each. The full dataset is split into three subsets, with starting dates being 01/01/2020, 01/01/2021, and 01/01/2022, respectively, and the ending date 01/01/2024. The subsets consist of 1006, 753, and 501 observations for each company, respectively.

We start with analysing the conditional autocorrelation function of an exemplary time series contained in the underlying dataset for which the proposed framework detected short-range serial dependence while the unconditional method did not; see Example~\ref{ex:ACF} for a similar analysis with synthetic data. More specifically, we focus on Catalent (CTLT) daily returns in the period 01/01/2023--01/01/2024. In Figure~\ref{plot:PARA_hist} we present the corresponding returns (left panel) and the lag plot (right panel) with $lag=1$ while in Figure~\ref{plot:ACF} we show the autocorrelation function, the autocorrelation function for the squared data, and the conditional autocorrelation function corresponding to the quantile split $p=0.01$ and $q=0.65$. In the returns chart, there are visible outliers which affect the estimation of autocorrelation measures. 

Indeed, in the autocorrelation function and the autocorrelation for the squared data, there are substantial spikes for relatively large values of lag (e.g. $lag=16$). This could be linked to the presence of outliers in the data, e.g. in mid April and early May 2023. Apart from the relatively big lags, one can see that neither the autocorrelation nor the autocorrelation for the squared sample detected some form of autodependence in the data. However, the significant dependence is reported by the conditional autocorrelation function; see e.g. the spike at $lag=1$. Also, the filtering mechanism embedded in the construction of $\rho_A$ is able to successfully reduce the impact of outliers.

\begin{figure}[H]
\begin{center}
\includegraphics[width=0.7\textwidth]{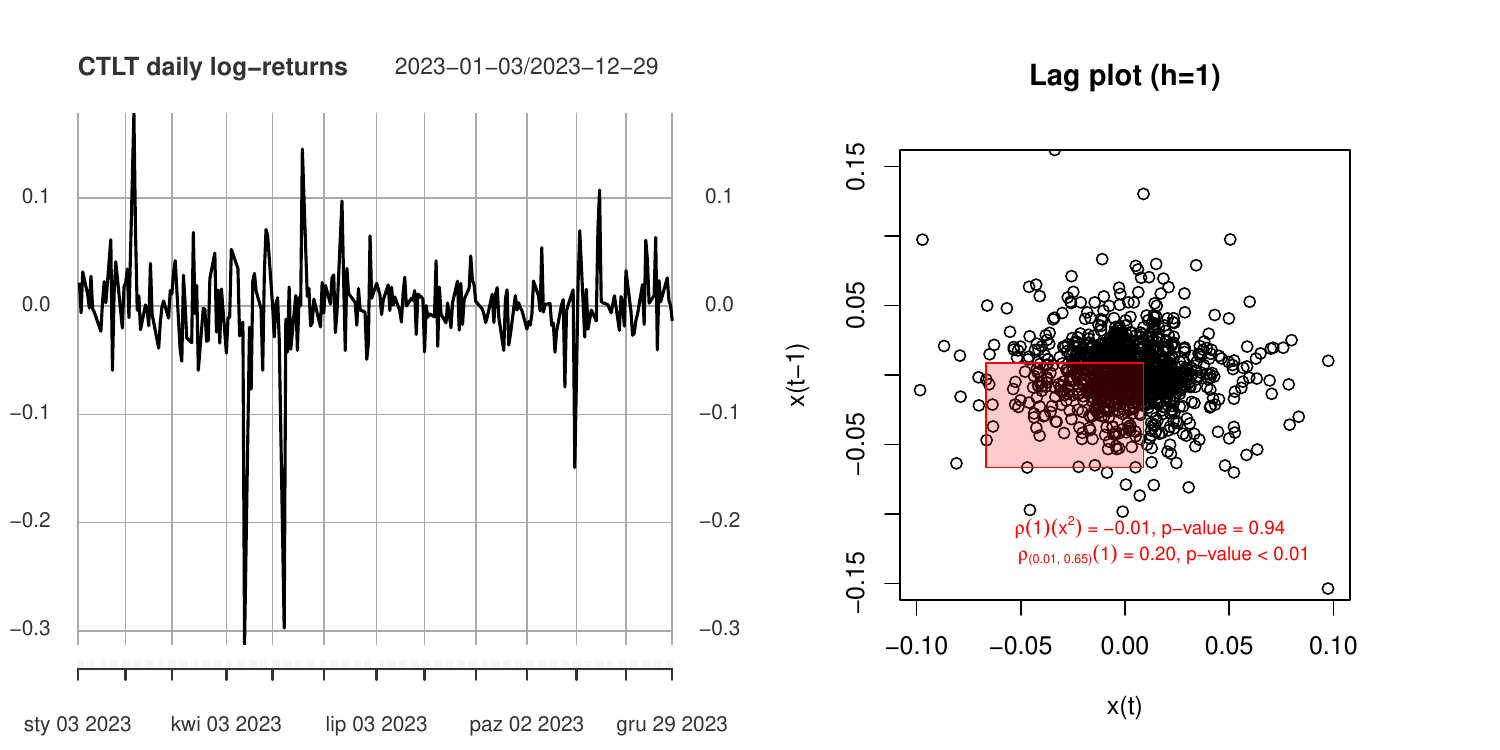}
\end{center}
\caption{Logarithmic daily rate of returns (left panel) and the corresponding lag plot (right panel) with $lag=1$ for CTLT in the period 01/01/2023-01/01/2024. The red region corresponds to the estimated set $A$ with $p=0.01$ and $q=0.65$.}\label{plot:PARA_hist}
\end{figure}

\begin{figure}[htp!]
\begin{center}
\includegraphics[width=0.32\textwidth]{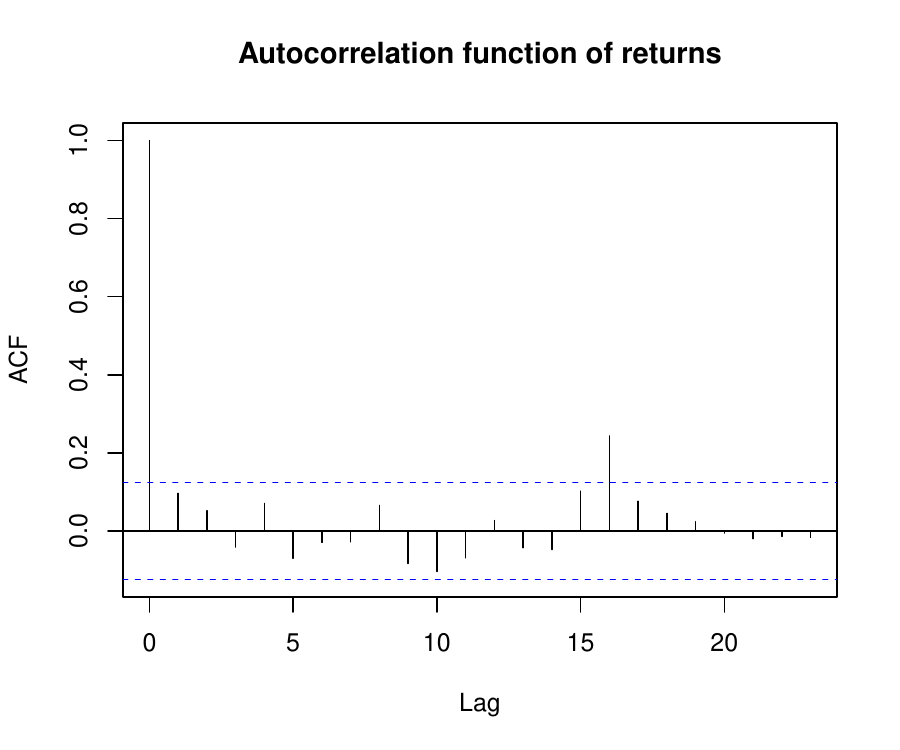}
\includegraphics[width=0.32\textwidth]{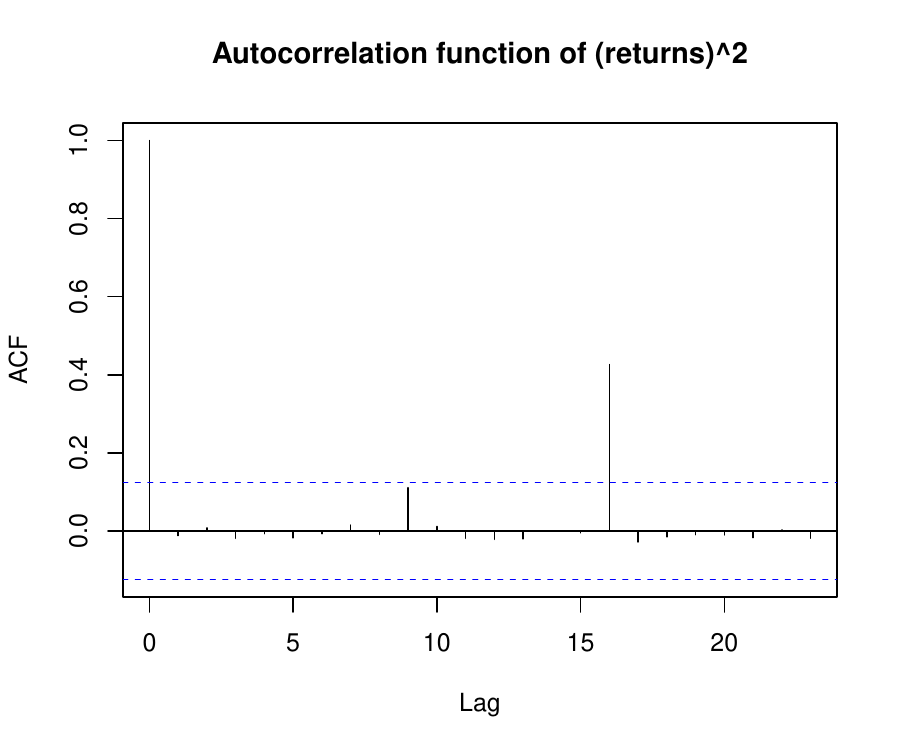}
\includegraphics[width=0.32\textwidth]{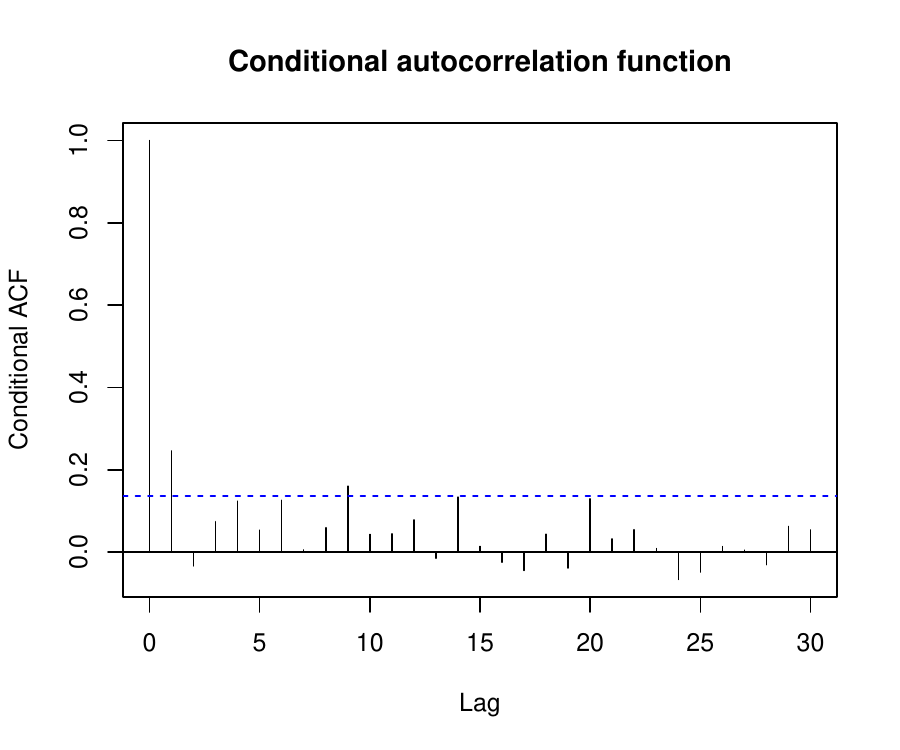}
\end{center}
\caption{Autocorrelation function (ACF) plots for CTLT stock price daily log-returns in the period 01/01/2023--01/01/2024. The left exhibit shows ACF applied to log-returns, the middle exhibit shows ACF applied to squared log-returns, while the right exhibit shows conditional ACF for quantile split $p =0.01$ and $q = 0.65$ applied to log-returns.}\label{plot:ACF}
\end{figure}

Next, we focus on unit lag $h=1$ and perform a more detailed analysis using the whole dataset. More specifically, we check the proportion of samples for which the conditional autocorrelation $\rho_A(1)$, with various choices of $0<p<q<1$, rejects the independence hypothesis. Also, we confront the results with the unconditional autocorrelation $\rho(1)$, and the unconditional autocorrelation applied to the squared sample denoted by $\rho(1)(x^2)$.
As already stated, a generic testing procedure described in Section~\ref{s4_intro} requires that the stationary distribution of the data is known in advance. Let us now discuss an extension of this framework which facilitates the relaxation of this assumption. In fact, it is enough to note that, under the stationarity assumption and assuming that the null hypothesis holds, the underlying data is i.i.d. Thus, instead of simulating Monte Carlo samples from the cumulative distribution function $F$, we apply a bootstrap procedure. More specifically, for each company and the corresponding rate of returns time series, we create 10,000 synthetic series by resampling with replacements from the original time series. Next, we compute the values of the considered statistics. i.e. conditional autocorrelation, autocorrelation, and autocorrelation for squared data, and define two-sided rejection regions by considering appropriate quantiles. Finally, for each method, we calculate the percentage of companies for which the value of the test static on the original sample lies in the corresponding rejection region.

The results are summarized in Table~\ref{table:PARA_scock}. Each sub-table corresponds to the different sample size, $n\in \{501,753,1006\}$, and a different type I error rate $\alpha_{I}\in\{0.05, 0.01\}$. For the conditional autocorrelation, we consider quantile splits $(p,q)\in \{(0.01,0.65), (0.01,0.75), (0.01,0.85), (0.15, 0.55), (0.55,0.85), (0.45,0.99)\}$.
In the table, we report the ratio of the samples with rejected null hypotheses (column ``Rejects''). Additionally, for the tests based on the conditional autocorrelations, we present the percentage of the samples for which the underlying test rejected the independence hypothesis while the test based on the autocorrelation for the squared data did not reject the null hypothesis (column ``U'').

\begin{table}[ht]
\centering
\begin{tabular}{|l|c|c|}
\hline
\multicolumn{3}{|c|}{$n=501$, $\alpha_{I}=0.05$}\\\hline
Statistic & Rejects (\%) & U (\%)  \\ \hline
$\rho(1)$  & $0.07$ &-\\
 $\rho(1)(x^2)$ &  $0.22$ &-\\ \hline
$\rho_{(0.15,0.55)}(1)$  & $0.05$ & 0.03\\
$\rho_{(0.55,0.85)}(1)$  & $0.05$ & 0.03\\
$\rho_{(0.01,0.65)}(1)$  & $0.25$ & 0.17\\
$\rho_{(0.01,0.75)}(1)$  & $0.25$ & 0.18\\
$\rho_{(0.01,0.85)}(1)$  & $0.19$ & 0.15\\
$\rho_{(0.45,0.99)}(1)$  & $0.11$ &0.10\\\hline
\end{tabular}
\,\,
\begin{tabular}{|l|c|c|}
\hline
\multicolumn{3}{|c|}{$n=753$, $\alpha_{I}=0.05$}\\\hline
Statistic & Rejects (\%) & U (\%) \\ \hline
$\rho(1)$  & $0.07$ &-\\
 $\rho(1)(x^2)$ &  $0.33$ &-\\ \hline
$\rho_{(0.15,0.55)}(1)$  & $0.05$ & 0.03\\
$\rho_{(0.55,0.85)}(1)$  & $0.05$ & 0.03\\
$\rho_{(0.01,0.65)}(1)$  & $0.35$ & 0.21\\
$\rho_{(0.01,0.75)}(1)$  & $0.34$ & 0.20\\
$\rho_{(0.01,0.85)}(1)$  & $0.22$ & 0.14\\
$\rho_{(0.45,0.99)}(1)$  & $0.16$ & 0.12\\\hline
\end{tabular}
\,\,
\begin{tabular}{|l|c|c|}
\hline
\multicolumn{3}{|c|}{$n=1006$, $\alpha_{I}=0.05$}\\\hline
Statistic & Rejects (\%) & U (\%)  \\ \hline
$\rho(1)$  & $0.51$ & -\\
 $\rho(1)(x^2)$ &  $0.93$ &-\\ \hline
$\rho_{(0.15,0.55)}(1)$  & $0.05$ & 0.00\\
$\rho_{(0.55,0.85)}(1)$  & $0.06$ & 0.00\\
$\rho_{(0.01,0.65)}(1)$  & $0.75$ & 0.05\\
$\rho_{(0.01,0.75)}(1)$  & $0.69$ & 0.03\\
$\rho_{(0.01,0.85)}(1)$  & $0.50$ & 0.03\\
$\rho_{(0.45,0.99)}(1)$  & $0.52$ & 0.03\\\hline
\end{tabular}

\vspace{0.3cm}

\begin{tabular}{|l|c|c|}
\hline
\multicolumn{3}{|c|}{$n=501$, $\alpha_{I}=0.01$}\\\hline
Statistic & Rejects (\%) & U (\%) \\ \hline
$\rho(1)$  & $0.02$ & -\\
 $\rho(1)(x^2)$ &  $0.10$ & -\\ \hline
$\rho_{(0.15,0.55)}(1)$  & $0.01$ & 0.00\\
$\rho_{(0.55,0.85)}(1)$  & $0.01$ & 0.01\\
$\rho_{(0.01,0.65)}(1)$  & $0.11$ & 0.09\\
$\rho_{(0.01,0.75)}(1)$  & $0.11$ & 0.09\\
$\rho_{(0.01,0.85)}(1)$  & $0.08$ & 0.07\\
$\rho_{(0.45,0.99)}(1)$  & $0.02$ & 0.03\\\hline
\end{tabular}
\,\,
\begin{tabular}{|l|c|c|}
\hline
\multicolumn{3}{|c|}{$n=753$, $\alpha_{I}=0.01$}\\\hline
Statistic & Rejects (\%) & U (\%) \\ \hline
$\rho(1)$  & $0.02$ & -\\
 $\rho(1)(x^2)$ &  $0.17$ & -\\ \hline
$\rho_{(0.15,0.55)}(1)$  & $0.01$ & 0.01\\
$\rho_{(0.55,0.85)}(1)$  & $0.00$ & 0.00\\
$\rho_{(0.01,0.65)}(1)$  & $0.15$ & 0.12\\
$\rho_{(0.01,0.75)}(1)$  & $0.15$ & 0.12\\
$\rho_{(0.01,0.85)}(1)$  & $0.10$ & 0.08\\
$\rho_{(0.45,0.99)}(1)$  & $0.02$ & 0.05\\\hline
\end{tabular}
\,\,
\begin{tabular}{|l|c|c|}
\hline
\multicolumn{3}{|c|}{$n=1006$, $\alpha_{I}=0.01$}\\\hline
Statistic & Rejects (\%) & U (\%) \\ \hline
$\rho(1)$  & $0.38$ & -\\
 $\rho(1)(x^2)$ &  $0.78$ &-\\ \hline
$\rho_{(0.15,0.55)}(1)$  & $0.01$ & 0.00\\
$\rho_{(0.55,0.85)}(1)$  & $0.02$ & 0.00\\
$\rho_{(0.01,0.65)}(1)$  & $0.47$ & 0.09\\
$\rho_{(0.01,0.75)}(1)$  & $0.43$ & 0.08\\
$\rho_{(0.01,0.85)}(1)$  & $0.24$ & 0.04\\
$\rho_{(0.45,0.99)}(1)$  & $0.27$ & 0.05\\\hline

\end{tabular}
\caption{
Percentage of samples (out of 498 in total) for which the underlying test statistic rejected the null hypothesis. In the tables, $n$ corresponds to the individual sample size while $\alpha_{I}$ is the confidence threshold, i.e. type I error rate. The column ``U`` presents the percentage of samples in which the test based on the corresponding conditional correlation rejected the null hypothesis while the test based on the (unconditional) correlation for squared data did not reject it.
}\label{table:PARA_scock}
\end{table}

Results presented in Table~\ref{table:PARA_scock} confirm the observations from the literature that the (unconditional) autocorrelation does not reject the independence hypothesis for the data, but the serial dependence is frequently reported by the (unconditional) autocorrelation for the squared data. However, it should be noted that even bigger rejection rates are reported by the conditional autocorrelation with $p=0.01$ and $q=0.65$. For example, for a two-year period ($n=501$) and the type I error rate equals 0.05, the test based on the conditional autocorrelation rejected the independence hypothesis for $25\%$ of the samples while the test based on the autocorrelation for squared data scored 22\% rejection rate. In fact, these two methods could be seen as complementary frameworks, as nearly 70\% of the dependence detections for the conditional autocorrelation were not detected by the autocorrelation for squares (statistic $U$ equals 17\%). This effect is also visible for a three-year period data ($n=753$), where the performance of $\rho(1)(x^2)$ is better than the conditional autocorrelation. Still, $\rho_{(0.01,0.65)}(1)$ was able to uniquely detect serial dependence in 5\% of the samples. This combined with the rejection rate for the autocorrelation for squared data (93\%) covers almost all considered samples.


\section*{Acknowledgements}
Marcin Pitera and Agnieszka Wy\l{}oma\'{n}ska acknowledge support from the National Science Centre, Poland, via project 2020/37/B/HS4/00120. Damian Jelito acknowledges support from the National Science Centre, Poland, via project 2020/37/B/ST1/00463.

\begin{footnotesize}
\bibliographystyle{agsm}
\bibliography{mybibliography}
\end{footnotesize}


\end{document}